\documentclass[conference]{IEEEtran}

\newif\ifpublish
\publishtrue

\newif\iflongversion
\longversiontrue

\newif\ifsbc
\sbcfalse

\usepackage{setspace}
\usepackage{xcolor}
\usepackage{amsfonts,amsmath,amsthm}
\usepackage{cryptocode}
\usepackage{multirow, tabu, makecell, graphics}
\usepackage{algorithm}
\usepackage[noend]{algpseudocode}
\usepackage{pifont}
\usepackage{xspace}
\usepackage{bm}
\usepackage[hidelinks]{hyperref}
\usepackage{listing}
\usepackage{listings}
\usepackage{cleveref}
\usepackage{pifont}
\usepackage{thmtools}
\usepackage{thm-restate}
\usepackage{subcaption}
\usepackage{booktabs}
\usepackage{multirow}
\usepackage{libertine}
\usepackage{enumitem}
\setitemize{noitemsep,topsep=0pt,parsep=0pt,partopsep=0pt}

\newcommand{\codelink}{
  \ifpublish
    \url{https://github.com/asonnino/mysticeti/tree/paper} (commit \texttt{96fd831})
  \else
    Code available but link omitted for blind review.
  \fi
}

\newcommand{\dslcodelink}{
  \ifpublish
    \url{https://github.com/MystenLabs/sui/pull/16436}
  \else
    Code available but link omitted for blind review.
  \fi
}

\newcommand{\para}[1]{\vskip 0.5em \noindent \textbf{#1.}}
\renewcommand{\IEEEiedlistdecl}{\settowidth{\labelindent}{0}}
\newtheorem{lemma}{Lemma}
\newtheorem{theorem}{Theorem}
\newtheorem{corollary}{Corollary}

\newcommand{\sysname}{\textsc{Mysticeti}\xspace}
\newcommand{\sysnamec}{\textsc{Mysticeti-C}\xspace}
\newcommand{\sysnamefpc}{\textsc{Mysticeti-FPC}\xspace}
\newcommand{\aquarium}{
  \ifpublish
    \unskip Sui\xspace
  \else
    \unskip Aquarium\xspace
  \fi
}
\newcommand{\realblockchain}{
  \ifpublish
    \unskip the Sui blockchain~\cite{sui}\xspace
  \else
    \unskip a major blockchain\xspace
  \fi
}
\newcommand{\GST}{\textsf{GST}}

\definecolor{eclipseStrings}{RGB}{42,0.0,255}
\definecolor{eclipseKeywords}{RGB}{127,0,85}
\colorlet{numb}{magenta!60!black}
\lstdefinelanguage{json}{
  basicstyle=\normalfont\ttfamily,
  commentstyle=\color{eclipseStrings},
  stringstyle=\color{eclipseKeywords},
  stepnumber=1,
  numbersep=8pt,
  showstringspaces=false,
  breaklines=true,
  string=[s]{"}{"},
  comment=[l]{:\ "},
  morecomment=[l]{:"},
  literate=
    *{0}{{{\color{numb}0}}}{1}
    {1}{{{\color{numb}1}}}{1}
    {2}{{{\color{numb}2}}}{1}
    {3}{{{\color{numb}3}}}{1}
    {4}{{{\color{numb}4}}}{1}
    {5}{{{\color{numb}5}}}{1}
    {6}{{{\color{numb}6}}}{1}
    {7}{{{\color{numb}7}}}{1}
    {8}{{{\color{numb}8}}}{1}
    {9}{{{\color{numb}9}}}{1}
}

\newcommand\YAMLcolonstyle{\color{red}\mdseries}
\newcommand\YAMLkeystyle{\color{black}\bfseries}
\newcommand\YAMLvaluestyle{\color{blue}\mdseries}
\makeatletter
\newcommand\language@yaml{yaml}
\expandafter\expandafter\expandafter\lstdefinelanguage
\expandafter{\language@yaml}
{
keywords={true,false,null,y,n},
keywordstyle=\color{darkgray}\bfseries,
basicstyle=\YAMLkeystyle,
sensitive=false,
comment=[l]{\#},
morecomment=[s]{/*}{*/},
commentstyle=\color{purple}\ttfamily,
stringstyle=\YAMLvaluestyle\ttfamily,
moredelim=[l][\color{orange}]{\&},
moredelim=[l][\color{magenta}]{*},
moredelim=**[il][\YAMLcolonstyle{:}\YAMLvaluestyle]{:},
morestring=[b]',
morestring=[b]",
literate =    {---}{{\ProcessThreeDashes}}3
{>}{{\textcolor{red}\textgreater}}1
{|}{{\textcolor{red}\textbar}}1
{\ -\ }{{\mdseries\ -\ }}3,
}
\lst@AddToHook{EveryLine}{\ifx\lst@language\language@yaml\YAMLkeystyle\fi}
\makeatother
\newcommand\ProcessThreeDashes{\llap{\color{cyan}\mdseries-{-}-}}

\begin{document}

\title{
  \sysname: Reaching the Latency Limits with Uncertified DAGs

  \ifsbc
    \vskip 0.1em
      {\LARGE (Accepted at NDSS'25, Adopted by Sui \& IOTA)}
  \fi
}

\ifpublish
  \author{
    \IEEEauthorblockN{
      Kushal Babel\IEEEauthorrefmark{1}\IEEEauthorrefmark{2},
      Andrey Chursin\IEEEauthorrefmark{3},
      George Danezis\IEEEauthorrefmark{3}\IEEEauthorrefmark{4},
      Anastasios Kichidis\IEEEauthorrefmark{3},
      Lefteris Kokoris-Kogias\IEEEauthorrefmark{3}\IEEEauthorrefmark{5}, \\
      Arun Koshy\IEEEauthorrefmark{3},
      Alberto Sonnino\IEEEauthorrefmark{3}\IEEEauthorrefmark{4},
      Mingwei Tian\IEEEauthorrefmark{3}
    }
    \IEEEauthorblockA{
      \\
      \IEEEauthorrefmark{1}Cornell Tech,
      \IEEEauthorrefmark{2}IC3,
      \IEEEauthorrefmark{3}Mysten Labs,
      \IEEEauthorrefmark{4}University College London (UCL),
      \IEEEauthorrefmark{5}IST Austria
    }
  }
\fi

\IEEEoverridecommandlockouts
\makeatletter\def\@IEEEpubidpullup{6.5\baselineskip}\makeatother
\IEEEpubid{
  \parbox{\columnwidth}{
    Network and Distributed System Security (NDSS) Symposium 2025\\
    24 - 28 February 2025, San Diego, CA, USA\\
    ISBN 979-8-9894372-8-3\\
    https://dx.doi.org/10.14722/ndss.2025.240929\\
    www.ndss-symposium.org
  }
  \hspace{\columnsep}\makebox[\columnwidth]{}
}

\thispagestyle{plain}
\pagestyle{plain}

\maketitle

\begin{abstract}
  We introduce \sysnamec, the first DAG-based Byzantine consensus protocol to achieve the lower bounds of latency of 3 message rounds.
Since \sysnamec is built over DAGs it also achieves high resource efficiency and censorship resistance.
\sysnamec achieves this latency improvement by avoiding explicit certification of the DAG blocks and by proposing a novel commit rule such that every block can be committed without delays, resulting in optimal latency in the steady state and under crash failures.
We further extend \sysnamec to \sysnamefpc, which incorporates a fast commit path that achieves even lower latency for transferring assets. Unlike prior fast commit path protocols, \sysnamefpc minimizes the number of signatures and messages by weaving the fast path transactions into the DAG. This frees up resources, which subsequently result in better performance.
We prove the safety and liveness in a Byzantine context. We evaluate both \sysname protocols and compare them with state-of-the-art consensus and fast path protocols to demonstrate their low latency and resource efficiency, as well as their more graceful degradation under crash failures. \sysnamec is the first Byzantine consensus protocol to achieve WAN latency of 0.5s for consensus commit while simultaneously maintaining state-of-the-art throughput of over 200k TPS. Finally, we report on integrating \sysnamec as the consensus protocol into \realblockchain, resulting in over 4x latency reduction.

\end{abstract}

\section{Introduction}
Several recent blockchains, such as Sui~\cite{sui,sui-lutris}, have adopted consensus protocols based on certified directed
acyclic graphs (DAG) of blocks~\cite{narwhal,bullshark,shoal,dag-rider,dumbo-ng,dispersedledger,sailfish,bbca-ledger,fino}. By design, these consensus protocols scale well in terms of throughput, with a performance of 100k tx/s of raw transactions and are robust against faults and network asynchrony~\cite{consensus-dos,narwhal}. This, however, comes at a high latency of around $2$-$3$ seconds, which can hinder user experience and prevent low-latency applications.

\para{\sysnamec: the power of uncertified DAGs}
Certified DAGs~\cite{dag-rider,narwhal}, where each vertex is delivered through consistent broadcast~\cite{cachin2011introduction}, have high latency for three main reasons: (1) the certification process requires multiple round-trips to broadcast each block between validators, get signatures, and re-broadcast certificates. This leads to higher latency than traditional consensus protocols~\cite{jolteon,diem,pbft}; (2) blocks commit on a ``per-wave'' basis, which means that only once every two rounds (for Bullshark~\cite{bullshark}) there is a chance to commit. Hence, some blocks have to wait for the wave to finish increasing the latency of transactions proposed by the block. This phenomenon is similar to committing big batches of $2f+1$ blocks. Finally, (3) since all certified blocks need to be signed by a supermajority of validators, signature generation and verification consume a large amount of CPU on each validator, which grows with the number of validators~\cite{li2023performance,chalkias2024fastcrypto}. This burden is particularly heavy for a crash-recovered validator that typically needs to verify thousands of signatures when trying to catch up with the rest. 

\begin{figure}[t]
    \centering
    \includegraphics[width=\columnwidth]{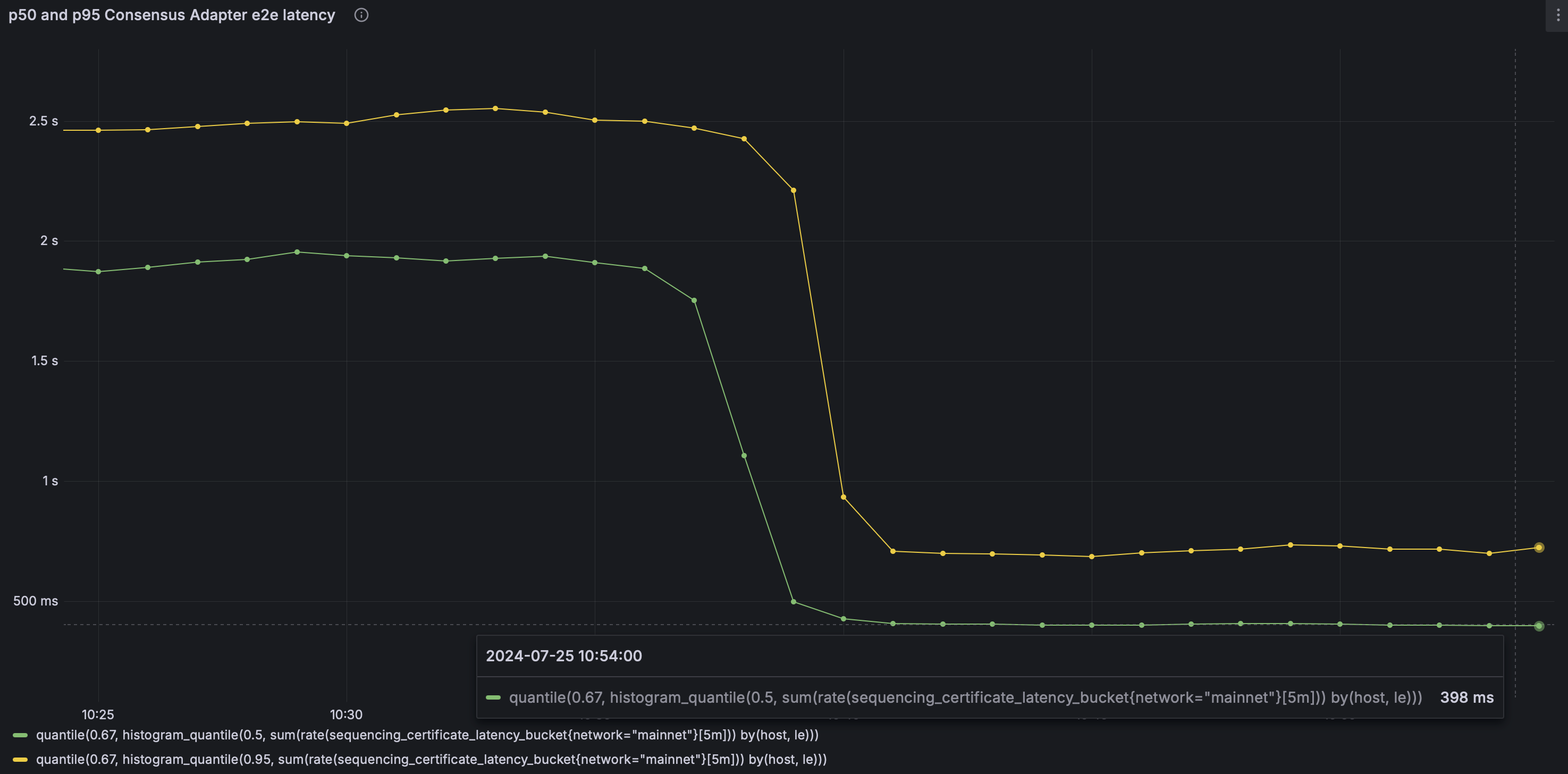}
    \caption{\footnotesize P50 latency of \realblockchain switching from Bullshark (1.9s) to \sysnamec (400ms) on 106 independently run validators}
    \label{fig:switch}
    \vspace{-0.5cm}
\end{figure}

These shortcomings come in stark contrast to the early protocols for BFT consensus, such as PBFT~\cite{pbft}, which require only 3 message delays to commit (instead of the $6$ in Bullshark) and facilitate the pipeline of proposals to commit one every round~\cite{kotla04high}. They, however, require a high number of authenticated messages to coordinate, which consumes a lot of resources and results in low throughput. Additionally, they are fragile to faults and implementation mistakes due to their complexity, especially the view-change sub-protocols.

This work presents \sysname, a family of DAG-based protocols allowing to safely commit distributed transactions in a Byzantine setting that
focuses on low-latency and low-CPU operation, achieving the best of both worlds.
\sysnamec is a consensus protocol based on a threshold logical clock~\cite{DBLP:journals/corr/abs-1907-07010} DAG of blocks,
that commits every block as early as it can be decided. \sysnamec solves all of the above challenges as (1) it does not require explicit certificates, committing blocks within the known lower bound~\cite{DBLP:journals/tdsc/MartinA06} of $3$ message rounds, (2) commits every block independently and does not need to wait for the wave to finish, and (3) requires a single signature generation and verification per block, minimizing the CPU overhead.

From a production readiness point of view, the protocol tolerates crash failures without any throughput degradation and minimal latency degradation. It uses a single message type, the signed block, and a single multi-cast transmission method between validators, making it easier to understand, implement, test, and maintain. \sysnamec has been adopted by \realblockchain that switched from the state-of-the-art Bullshark~\cite{bullshark} to \sysnamec. \Cref{fig:switch} shows the 80\% latency reduction (from 1.9s to 400 ms) that happened at the moment of the deployment on a 106 validator network.

\para{\sysnamefpc: supporting consensusless transactions}
The power of uncertified DAGs is not limited to consensus protocols.
This work generalizes \sysnamec to apply uncertified DAGs to BFT systems that process transactions without or before reaching consensus, such as in FastPay~\cite{fastpay}, Zef~\cite{DBLP:journals/corr/abs-2201-05671}, Astro~\cite{astro}, and Sui~\cite{sui-lutris}.
These systems use reliable broadcast instead of consensus to commit transactions that only access state
controlled by a single party.

The only operating protocol of this kind is Sui Lutris~\cite{sui-lutris}, which powers
the open source Sui blockchain (Linera~\cite{linera} is under development). Sui combines a consensusless ``fast'' path with a black-box certified DAG consensus. This composition is generic and leads to low latencies for fast-path transactions. But it also leads to
(1) increased latencies for transactions requiring the consensus path and overall increased sync latency due to a separate post-consensus checkpoint mechanism,
and (2) additional signature generation and
verification for transactions to be certified separately. The latter means that the validator's CPU is largely
devoted to performing cryptographic operations rather than executing transactions.
To alleviate these challenges, we co-design with \sysnamec a fast path-enabled version called \sysnamefpc, leading to very low-latency commits without the need to generate an explicit certificate for each transaction.
This new design inherits the benefits of lower latency and lower CPU utilization.

\para{Contributions} We make the following contributions:
\begin{itemize}
    \item We present \sysnamec, a DAG-based Byzantine consensus algorithm and its proofs of safety and liveness. Notably, it implements a commit rule where every single block can be directly committed, significantly reducing latency even when failures occur. We show it has a low commit latency and exceeds the throughput of Narwhal-based consensus. \sysnamec is already powering \realblockchain with more than \$1.5B of value under management and 1M Daily Active Accounts.

    \item We also present \sysnamefpc that offers feature parity with Sui Lutris~\cite{sui-lutris}, that is, both a fast path and a consensus path, as well as safe checkpointing and epoch close mechanisms. We show that \sysnamefpc has a fast path latency comparable with Zef~\cite{DBLP:journals/corr/abs-2201-05671} and Fastpay~\cite{fastpay} but higher throughput thanks to lower CPU utilization and batching.

    \item We implement and evaluate both protocols on a wide-area network. We show their performance is superior to certified DAG-based designs both in consensus and consensusless modes due to the need for fewer messages and lower CPU overheads. We also report the experiences and performance benefits of integrating \sysnamec into a production blockchain.
\end{itemize}

\section{Overview} \label{sec:overview}
This paper presents the design of the \sysname protocols, a pair of Byzantine Fault Tolerant (BFT) protocols based on Directed Acyclic Graphs (DAGs) that aim to achieve high performance in a partially synchronous network. \sysnamec is a low-latency consensus protocol that commits multiple blocks per round, while \sysnamefpc extends \sysnamec with a fast path for transactions that do not require consensus.

\subsection{System model, goals, and assumptions}
We consider a message-passing system where, in each epoch, $n=3f+1$ validators process transactions using the \sysname protocols.
In every epoch, a computationally bound adversary can statically corrupt an unknown set of up to $f$ validators. We call these validators \emph{Byzantine} and they can deviate from the protocol arbitrarily. The remaining validators (at least $2f+1$) are \emph{honest} and follow the protocol faithfully.

For the description of the protocol, we assume that links between
honest parties are reliable and authenticated. That is, all messages among honest parties eventually arrive and a receiver can verify the sender's identity.
The adversary is computationally bound hence the usual security properties of cryptographic hash functions, digital signatures, and other cryptographic primitives hold. Under these assumptions, \Cref{sec:security} shows that the \sysname protocols are safe, in that, no two correct validators commit inconsistent transactions.

Validators communicate over a partially synchronous network. There exists a time called Global Stabilization Time ($\GST$) and a finite time bound $\Delta$, such that any message sent by a party at time $x$ is guaranteed to arrive by time $\Delta + \max\{\GST, x\}$. Within periods of synchrony (after GST) the \sysname protocols are also live in that they are guaranteed to commit transactions from correct validators.

Following prior work~\cite{dag-rider,bullshark,narwhal} we focus on byzantine atomic broadcast for \sysname. Additionally for \sysnamefpc, we show that the fast-path transactions sub-protocol satisfies reliable broadcast within an epoch~\cite{sui-lutris}, but allows for recovery of equivocating objects across epochs without losing safety at the epoch boundaries.

More formally, each validator $v_k$ broadcasts messages by calling
$\textit{r\_bcast}_k(m,q)$, where $m$ is a message and $q \in \mathbb{N}$
is a sequence number.
Every validator $v_i$ has an output $\textit{r\_deliver}_i(m,q,v_k)$,
where $m$ is a message, $q$ is a sequence number, and $v_k$ is the identity of the validator
that called the corresponding $\textit{r\_bcast}_k(m,q)$.
The reliable broadcast abstraction guarantees the following properties:
\begin{itemize}
      \item \textbf{Agreement:} If an honest validator $v_i$ outputs $\textit{r\_deliver}_i(m,q,v_k)$, then every other honest validator $v_j$ eventually outputs $\textit{r\_deliver}_j(m,q,v_k)$.
      \item \textbf{Integrity:} For each sequence number $q \in \mathbb{N}$ and validator $v_k$, an honest validator $v_i$ outputs $\textit{r\_deliver}_i(m,q,v_k)$ at most once regardless of~$m$.
      \item \textbf{Validity:} If an honest validator $v_k$ calls $\textit{r\_bcast}_k(m,q)$, then every honest validator $v_i$ eventually outputs $\textit{r\_deliver}_i(m,q,v_k)$.
\end{itemize}
Additionally, for byzantine atomic broadcast, each honest
validator $v_i$ can call $\textit{a\_bcast}_i(m,q)$ and output
$\textit{a\_deliver}_i(m,q,v_k)$.
A byzantine atomic broadcast protocol satisfies reliable broadcast
(agreement, integrity, and validity) as well as:
\begin{itemize}
      \item \textbf{Total order:} If an honest validator $v_i$ outputs $a\_deliver_i(m,q,v_k)$ before $a\_deliver_i(m',q',v_k')$, then no honest party $v_j$ outputs $a\_deliver_j(m',q',v_k')$ before $a\_deliver_j(m,q,v_k)$.
\end{itemize}

Finally, most prior work on consensusless transactions defines properties as if the protocol runs in a single epoch. This setting is unrealistic as it cannot accommodate recovering from equivocation, which is a common benign event for non-expert users. To this end, we extend all the protocols to also take as a parameter the epoch number and all properties should hold within a single epoch.
Fortunately, the definition of reliable broadcast allows the recovery of liveness for blocked sequence numbers that are equivocated inside an epoch. Thus, we define equivocation tolerance for consesusless transactions as follows:
\begin{itemize}
      \item \textbf{Equivocation tolerance:} If a validator $v_k$ concurrently called $\textit{r\_bcast}_k(m,q,e)$ and  $\textit{r\_bcast}_k(m',q,e)$ with $m \neq m'$ then the rest of the validators either $\textit{r\_deliver}_i(m,q,v_k,e)$, or $\textit{r\_deliver}_i(m',q,v_k,e)$, or there is a subsequent epoch $e'>e$ where $v_k$ is honest, calls $\textit{r\_bcast}_k(m'',q,e')$ and all honest validators $\textit{r\_deliver}_i(m'',q,v_k,e')$,
\end{itemize}

\subsection{Intuition behind the \sysname design}
\sysname aims to push the latency boundaries of state machine replication in DAG-based blockchains.
Achieving BFT consensus typically necessitates at least three message delays~\cite{pbft}\footnote{
      While some protocols, such as Zyzzyva~\cite{zyzzyva}, operate under optimistic assumptions, they often prove fragile in scenarios of asynchrony or faults~\cite{narwhal,consensus-dos}. Moreover, they are unsuitable for the blockchain environment, characterized by a multitude of unreliable nodes wielding a minor fraction of the total voting power.
}. This underscores the inherent latency sub-optimality of Narwhal~\cite{narwhal}, that implements consensus (at least 3 message delays) on certified DAG blocks, when the block certification itself adds a further 3 message delays. Consequently, the first design challenge for \sysname is to manage equivocation and ensure data availability~\cite{cohen2023proof}, without relying on pre-certification of individual blocks.

Moreover, even if we overcome this initial challenge, committing only one block every three messages falls short of the performance potential inherent in DAG-based consensus, which thrives on processing $O(n)$ blocks per round, one per validator, to fully utilize network resources. Therefore, a key objective for \sysname is to maximize block commitments per round to align system tail latency closely with the three-message delay. However, achieving this presents a more formidable challenge. Unlike traditional methods that rely on the recursive and elegant commit rules found in DAG-based consensus protocols~\cite{dag-rider, narwhal, bullshark, dumbo-ng,dispersedledger}, our approach cannot afford to require sufficient distance between two potential candidate blocks on the DAG to prevent conflicting decisions among validators with divergent sub-DAG views. Implementing such protocols would require at least one gap round, raising the latency to a minimum of four delays.

\sysname is not just a consensus protocol but a class of protocols facilitating state machine replication. For now, we only focused on the consensus protocol \sysnamec, but \cref{sec:fastpath} extends it to protocols for consensusless agreement with \sysnamefpc. The core contribution of \sysnamefpc to prior work is that it is co-designed with \sysnamec instead of being a separate path like in Sui~\cite{sui-lutris}. This allows us to avoid the need for generating a majority-signed certificate per transaction, freeing a significant amount of network and CPU resources to be used for actual transactions instead of generating and verifying certificates~\cite{li2023performance,chalkias2024fastcrypto}.

Given ourexperience of deploying DAG-based consensus protocols~\cite{sui-lutris}, there are some design challenges that relate to engineering. Bullshark~\cite{bullshark} requires separate sub-protocols for managing individual block certification, for exchanging certified blocks, and for managing the communication of metadata between nodes. The challenge with \sysname is to design a protocol that has a single message type, the signed block, and a single network primitive, by which each block is multi-cast to all other correct validators.

A final point of focus inspired by our deployment is that crash-faults and struggling nodes are a common occurrence and not an exception. This is why we have designed \sysname to be able to tolerate crash-faults with as little performance degradation as possible.

\subsection{The structure of the \sysname DAG} \label{sec:dag}
We present the structure of the \sysname DAG. Its main goal is to build an uncertified DAG protocol that provides the same guarantees as a certified DAG.

The \sysname protocols operate in a sequence of logical \emph{rounds}.  For every round, each honest validator proposes a unique signed \emph{block}; Byzantine validators may attempt to equivocate by sending multiple distinct blocks to different parties or no block. During a round, validators receive transactions from users and blocks from other validators and use them as part of their proposed blocks.
A block includes \emph{references to blocks} from prior rounds, always starting from their most recent block, alongside \emph{fresh transactions} not yet incorporated indirectly in preceding blocks.
Once a block contains references to at least $2f+1$ blocks from the previous round, the validator signs it and sends it to other validators.

Clients submit transactions to a validator, who subsequently incorporates them into their blocks. In the event that a transaction fails to become finalized within a specified time frame, the client selects an alternative validator for resubmission.

\para{Block correctness}
A block should include at a minimum (1) the author $A$ of the block and their signature on the block contents, (2) a round number $r$, (3) a list of transactions, and (4) at least $2f+1$ distinct hashes of blocks from the previous round, along potentially others from all previous rounds. By convention, the first hash must be to the previous block of $A$\footnote{
      This rule also helps to guarantee the safety of fast pah transactions upon epoch change (\Cref{sec:epoch-change}).
}. We index each block by the triplet $B \equiv (A, r, h)$, comprised of the author $A$, the round $r$, and the hash $h$ of the block contents.
A block is valid if (1) the signature is valid and $A$ is part of the validator set, and (2) all hashes point to distinct valid blocks from previous rounds, the first block links to a block from $A$, and within the sequence of past blocks, there are $2f+1$ blocks from the previous round $r-1$.

\para{Identifying DAG patterns} \label{sec:dag-patterns}
We say that a block $B'$ \emph{supports} a past block $B \equiv (A, r, h)$ if, in the depth-first search performed starting at $B'$ and recursively following all blocks in the sequence of blocks hashed, block $B$ is the first block encountered for validator $A$ at round $r$.
As \Cref{fig:support} illustrates, a block $(A_3, r+2, \cdot)$ (green) may reference blocks $(A_2, r+1, \cdot)$ and $(A_3, r+1, \cdot)$ from different validators that respectively support block $(A_3, r, L_r)$ (blue) and the equivocating block $(A_3, r, L_r')$ (red). At most one of these equivocating blocks can gather support from $2f+1$ validators.

\begin{figure}[t]
      \centering
      \includegraphics[scale=0.45]{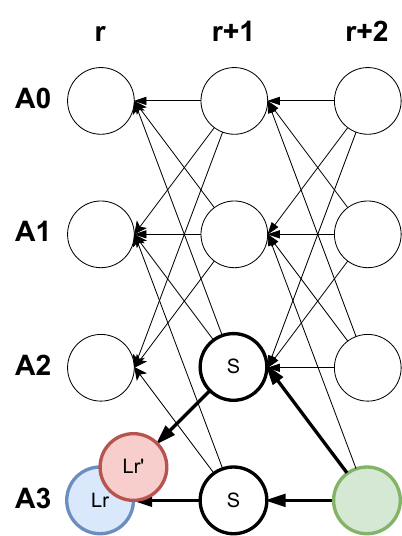}
      \caption{\footnotesize Block $(A_3, r+2, \cdot)$ (green) may reference blocks from different validators that support both $(A_3, r, L_r)$ (blue) and $(A_3, r, L_r')$ (red) equivocating blocks. If any of the blocks gathers $2f+1$ support, it will be certified, and we show that at most one may do so.}
      \vspace{-0.5cm}

      \label{fig:support}
\end{figure}

\sysnamec (\Cref{sec:consensus}) and \sysnamefpc (\Cref{sec:fastpath}) operate by interpreting the structure of the DAG to reach decisions using a single type of message, the block.
They mainly operate by identifying the following two patterns:
\begin{enumerate}
      \item The \emph{skip pattern}, illustrated by \Cref{fig:dag-patterns} (left), where at least $2f+1$ blocks at round $r+1$ \emph{do not} support a block $(A, r, h)$. Note that there may be multiple or no proposal for the slot. The skip pattern is identified if for all proposals, we observe 2f+1 subsequent blocks that do not support it (or support no proposal).
      \item The \emph{certificate pattern}, illustrated by \Cref{fig:dag-patterns} (right), where at least $2f+1$ blocks at round $r+1$  \emph{support} a block $B \equiv (A, r, h)$. We then say that $B$ is \emph{certified}. Any subsequent block (illustrated at $r+2$) that contains in its history such a pattern is called a \emph{certificate} for the block $B$.
\end{enumerate}

Using these patterns, we obtain certificates implicitly by interpreting the DAG, and the certification guarantees are identical to Narwhal~\cite{narwhal}. That is, a certified block ($2f+1$ support) is available and no other certified block may exist for the same spot $(A, r)$. This counter intuitively means that even if $A$ equivocates and one of its blocks is certified, we process it as being correct -- despite the self evident Byzantine behavior. This does not constitute a problem as we only commit blocks that belong to the implicitly certified part of the DAG. We also note that a skip pattern guarantees that a certificate will never exist for a block, and thus it will never be part of the implicitly certified DAG and can be safely skipped.

\begin{figure}[t]
      \begin{subfigure}[t]{0.23\textwidth}
            \centering
            \includegraphics[scale=0.45]{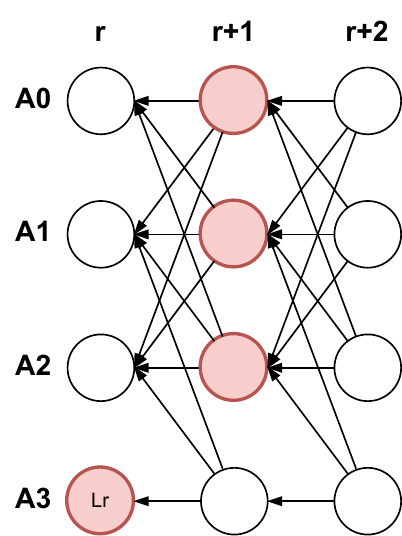}
            \caption{
                  \footnotesize Illustration of \emph{skip} pattern, blocks $(A_0, r+1, \cdot), (A_1, r+1, \cdot), (A_2, r+1, \cdot)$ do not support $(A_3, r, L_r)$.
            }
            \label{fig:skip-pattern}
      \end{subfigure}
      \hfill
      \begin{subfigure}[t]{0.23\textwidth}
            \centering
            \includegraphics[scale=0.45]{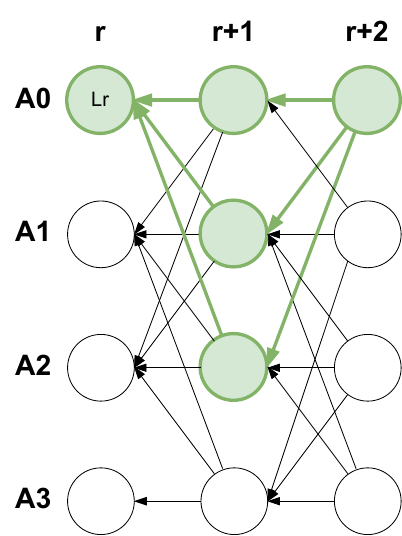}
            \caption{
                  \footnotesize Illustration of \emph{certificate} pattern, block $(A_0, r+2, \cdot)$ is a certificate for $(A_0, r, L_r)$.
            }
      \end{subfigure}
      \caption{
            \footnotesize Illustration of main DAG patterns identified by validators.
      }
      \vspace{-0.5cm}
      \label{fig:dag-patterns}
\end{figure}

\para{Liveness intuition} \label{sec:dag-liveness}
Since we are not using randomization, we need to rely on timeouts for liveness. Although every blocks has the potential of being committed directly in 3 message delays we cannot provide liveness for all of them through timeouts, as this would allow Byzantine validators to slow down the DAG to the point that every round would move at the speed of the timeout instead of network speed.

Instead we only provide guaranteed liveness after GST for one block per round\footnote{
      This can be extended to more blocks but it increases the chance that the adversary controls one block causing a full delay for the round.
}. We deem this block as the primary block of the round $r$ and require that validators at $r+1$ wait a timeout for it to arrive before disseminating their blocks. Additionally, if the block is in the view of a validator at $r+1$ we further require the validator to wait another timeout for $r+2$ or until there are $2f+1$ votes for the primary block of $r$. This guarantees the existence of a certificate over an honest primary block after GST and provides liveness for \sysnamec.

\section{The \sysnamec Consensus Protocol} \label{sec:consensus}
\sysnamec is the first DAG-based consensus protocol that decides blocks in 3 message delays. It achieves this through foregoing an explicit certification of the blocks and through treating every block as a first-class block that can be proposed and decided directly. Additionally, \sysnamec is able to instantly identify and exclude crashed validators, the most frequent failure case in blockchains in the wild.

\begin{figure*}[t]
    \centering
    \vspace{-0.4cm}
    \begin{subfigure}[t]{0.3\textwidth}
        \centering
        \includegraphics[width=\textwidth]{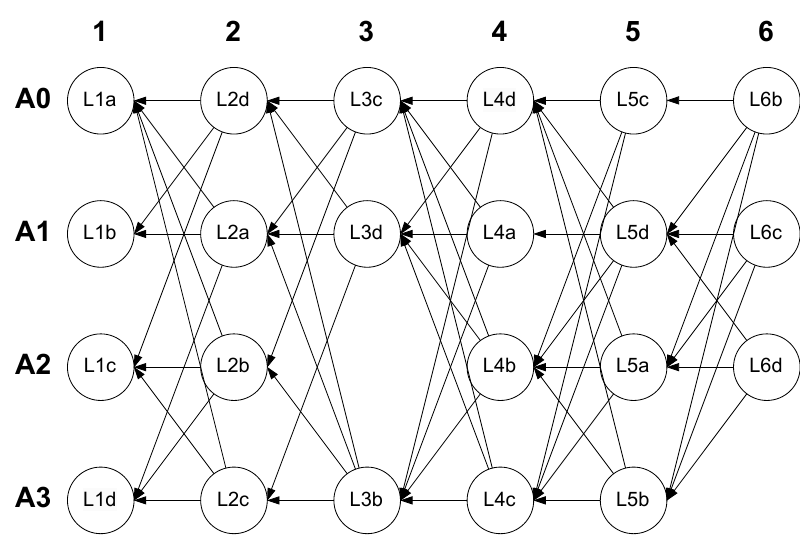}
        \caption{\footnotesize All proposers are initially \textsf{undecided}.}
        \label{fig:example-1}
    \end{subfigure}
    \hfill
    \begin{subfigure}[t]{0.3\textwidth}
        \centering
        \includegraphics[width=\textwidth]{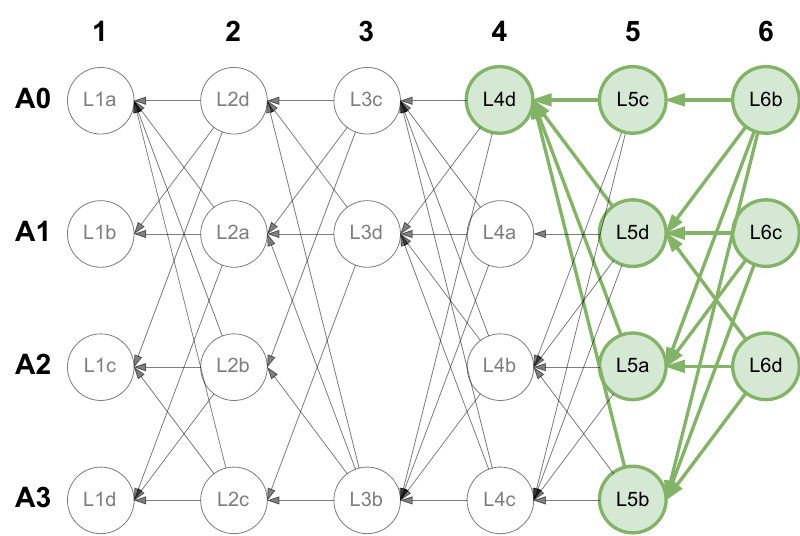}
        \caption{\footnotesize Direct decision rule: L4d is \textsf{to-commit}.}
        \label{fig:example-2}
    \end{subfigure}
    \hfill
    \begin{subfigure}[t]{0.3\textwidth}
        \centering
        \includegraphics[width=\textwidth]{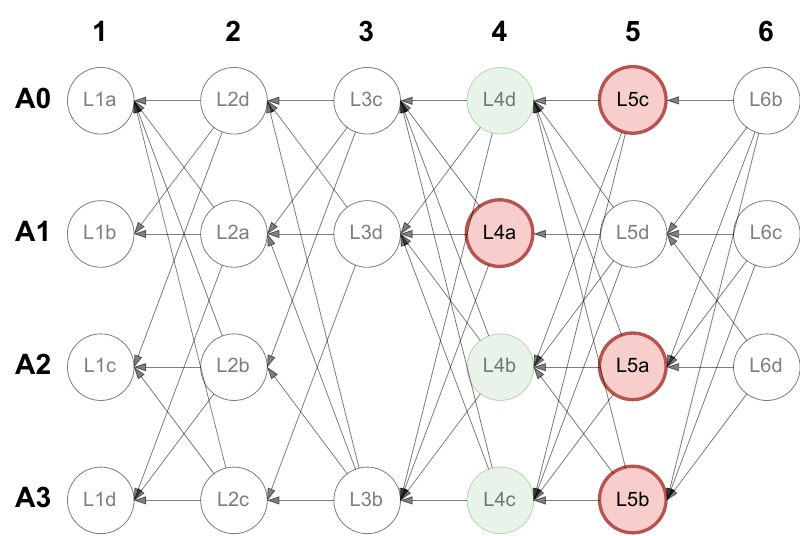}
        \caption{\footnotesize Direct decision rule: L4a is \textsf{to-skip}.}
        \label{fig:example-3}
    \end{subfigure}
    \hfill
    \begin{subfigure}[t]{0.3\textwidth}
        \centering
        \includegraphics[width=\textwidth]{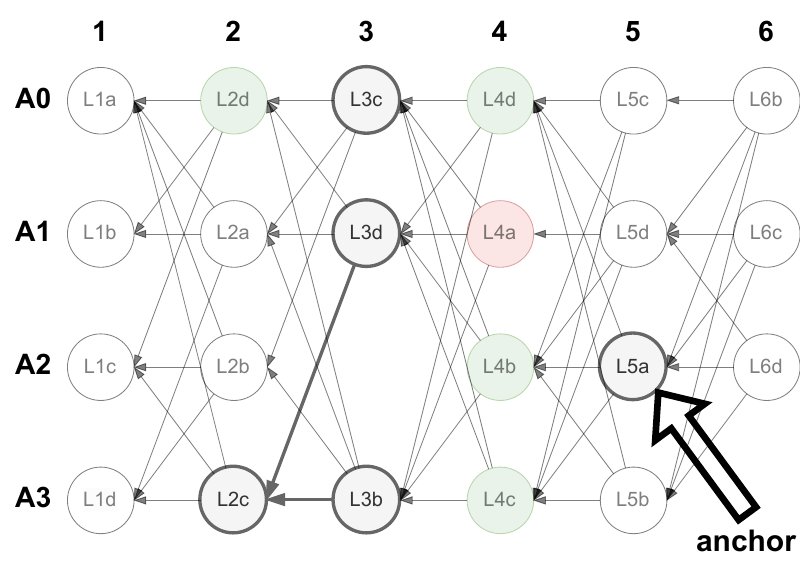}
        \caption{\footnotesize Indirect decision rule: L2c is \textsf{undecided}. Its anchor (L5a) is \textsf{undecided}, we cannot determine the status of L2c yet.}
        \label{fig:example-4}
    \end{subfigure}
    \hfill
    \begin{subfigure}[t]{0.3\textwidth}
        \centering
        \includegraphics[width=\textwidth]{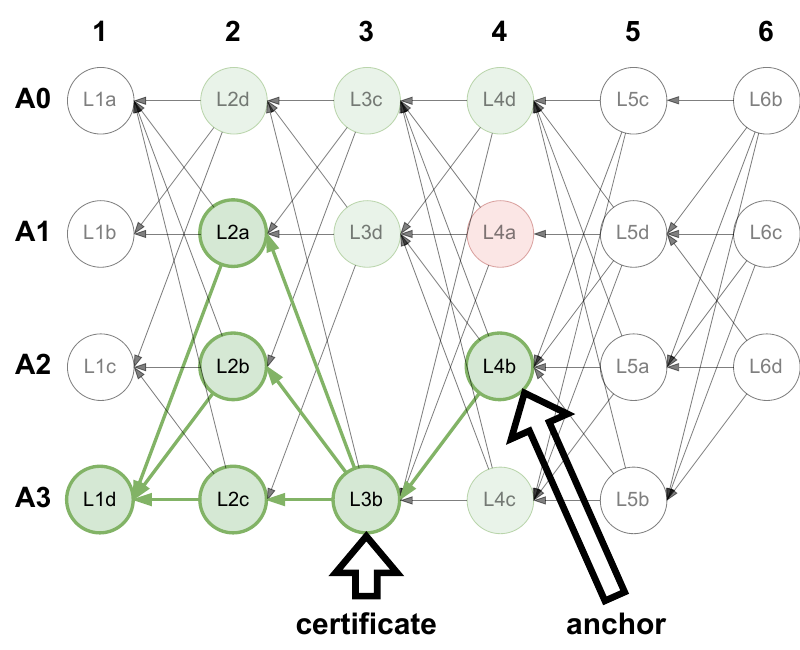}
        \caption{\footnotesize Indirect decision rule: L1d is \textsf{to-commit}. Its anchor (L4b) is \textsf{to-commit} and there's a certified link from L4b to L1d.}
        \label{fig:example-5}
    \end{subfigure}
    \hfill
    \begin{subfigure}[t]{0.3\textwidth}
        \centering
        \includegraphics[width=\textwidth]{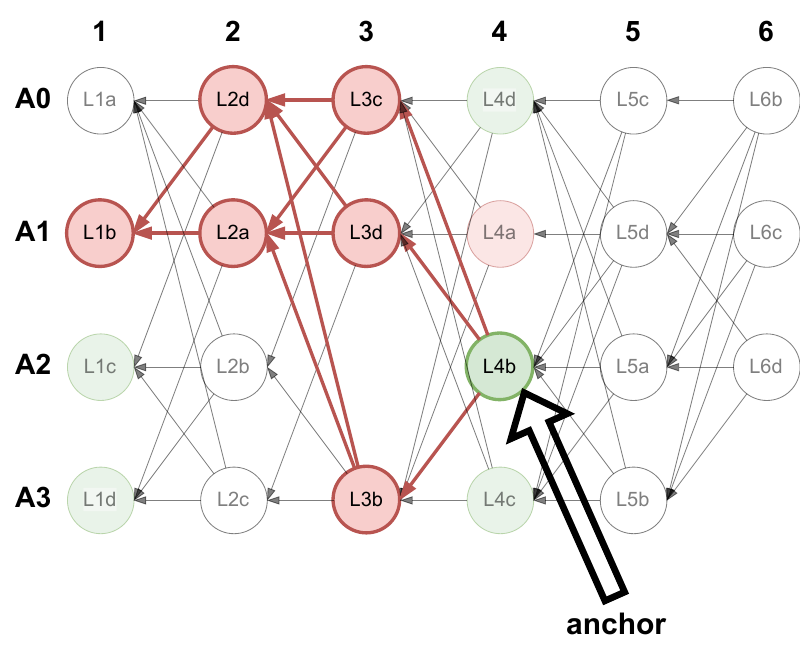}
        \caption{\footnotesize Indirect decision rule: L1b is \textsf{to-skip}. Its anchor (L4b) is \textsf{to-commit} and no certified link (only two links in round 2) from L4b to L1b.}
    \end{subfigure}
    \hfill

    \caption{\footnotesize Example application of the \sysnamec decision rule with four validators (A0, A1, A2, A3) and four proposer slots per round.}
    \vspace{-0.5cm}

    \label{fig:example-6}
    \label{fig:example}
\end{figure*}

\subsection{Proposer slots}
\sysnamec introduces the concept of \emph{proposer slot}. A proposer slot represents a tuple (validator, round) and can be either empty or contain the validator's proposal for the respective round. For instance, in Bullshark~\cite{bullshark}, there is a single proposer every two rounds, which results in higher latencies. Unfortunately, it is not trivial to increase the number of slots, as the commit rule of Bullshark relies on the fact that every proposer slot has a link to every other proposer slot, something that is not possible even if there is a single proposer per round, let alone $n$.

We overcome this challenge by introducing multiple \emph{states} for each proposer slot, namely: \textsf{to-commit}, \textsf{to-skip}, or \textsf{undecided}.
The \textsf{to-commit} state is the equivalent of the \textsf{decided} state that already exists in the prior work. The most important state is the \textsf{undecided},
which forces all subsequent proposer slots to wait, mitigating the risk of non-deterministic commitments due to network asynchrony without the need for a buffer round as prior work~\cite{dag-rider,narwhal,bullshark,dumbo-ng}. Finally, the \textsf{to-skip} state allows to exclude proposer slots assigned to crashed validators, thus allowing the subsequent slots to commit.

The number of proposer slots instantiated per round can be configured but for systems with few faults it can be set to $n$ so that every block has a chance to commit in $3$ steps. It can also be dynamically adjusted based on the network conditions, following a similar deterministic approach to HammerHead~\cite{tsimos2023hammerhead} (see
\iflongversion
    \Cref{sec:example}
\else
    the long version of the paper~\cite{mysticeti-long}
\fi
).
Initially, we establish a deterministic total order among all pending proposer slots, aligning with the round ordering. Within a single round, the ordering may either remain fixed or change per round (e.g., round robin).
\Cref{fig:example} illustrates an example of a \sysname DAG with four validators, (A0, A1, A2, A3), four slots per round, and a potential proposer slot ordering represented as (L1a, L1b, L1c, L1d) and (L2a, L2b, L2c, L2d) for the first and second rounds, respectively. This order resembles a FIFO queue.

As discussed in \Cref{sec:dag-liveness}, validators await the proposal from the primary validator assigned to the first proposer slot of round $r$ for up to a predetermined delay $\Delta$ before generating their own proposal for round $r+1$. \Cref{sec:security} shows that this delay ensures the liveness of the protocol.

\subsection{The \sysnamec decision rule}
This section describes the decision rule of \sysnamec leveraging an example protocol run. \Cref{sec:algorithms} provides detailed algorithms.
As illustrated by \Cref{fig:example-1}, all proposer slots are initially in the \textsf{undecided} state. The end goal of \sysnamec is to mark all proposer slots as either \textsf{to-commit} or \textsf{to-skip} by detecting the DAG patterns presented in \Cref{sec:dag-patterns}. The \sysnamec decision rule operates in three steps:

\para{Step 1: Direct decision rule}
Starting with the latest proposer slot (L6d in \Cref{fig:example}), the validator applies the following \emph{direct decision rule} to attempt to determine the status of the slot.
The validator marks a slot as \textsf{to-commit} if it observes $2f+1$ \emph{commit patterns} for that slot, that is, if it accumulates $2f+1$ distinct implicit certificate blocks for it (see \Cref{sec:dag}). This is the \textbf{first key design point} for lowering the latency as we certify blocks while constructing the DAG by interpreting \emph{certificate patterns}.

\Cref{fig:example-2} illustrates the direct decision rule applied to L4d, which is marked as \textsf{to-commit} in just 3 messages due to the presence of $2f+1$ commit patterns.
The first message delay is the proposal block; the second message delay is the block(s) supporting and voting/certification; and the third message delay is the block(s) certifying serving as acknowledgment/commitment.
The direct decision rule marks a slot as \textsf{to-skip} if it observes a \emph{skip pattern} for that slot. That is for any proposal for the slot (there may be multiple due to potential equivocation) it observes $2f+1$ blocks that do not support it or support no proposal.
\Cref{fig:example-3} demonstrates the direct decision rule applied to L4a, which is marked as \textsf{to-skip} due to the presence of a skip pattern.

Promptly marking slots as \textsf{to-skip} is the \textbf{second key design point} that contributes to the reduction of \textsf{undecided} slots following crash-failures and allows \sysnamec to tolerate crash-faults virtually for free.

If the direct decision rule fails to mark a slot as either \textsf{to-commit} or \textsf{to-skip}, the slot remains \textsf{undecided} and the validator resorts to the \emph{indirect decision rule} presented in step 2 below. During normal operations, however, we expect the direct decision rule to succeed and to only resort to the indirect decision rule during periods of asynchrony or under attacks.

\para{Step 2: Indirect decision rule}
If the direct decision rule fails to determine the slot, the validator resorts to the indirect decision rule to attempt to reach a decision for the slot.
This rule operates in two stages. It initially searches for an \emph{anchor}, which is defined as the first slot with the round number $(r' > r + 2)$ that is already marked as either \textsf{undecided} or \textsf{to-commit}\footnote{This section assumes a fixed distance of 3 rounds between a proposer slot which is the minimum secure distance. \Cref{sec:algorithms} generalize this rule to a variable distance and discusses its tradeoffs.}.
\Cref{fig:example-4} and \Cref{fig:example-5} respectively illustrate the anchor of L2c (marked as \textsf{undecided}) and the anchor of L1d (marked as \textsf{to-commit}).

If the anchor is marked as \textsf{undecided} the validator marks the slot as \textsf{undecided} (\Cref{fig:example-4}).
Conversely, if the anchor is marked as \textsf{to-commit}, the validator marks the slot either as \textsf{to-commit} if the anchor causally references a certificate pattern over the slot or as \textsf{to-skip} in the absence of a certificate pattern.
\Cref{fig:example-5} illustrates the indirect decision rule applied to L1d, which is marked as \textsf{to-commit} due to the presence of a certificate pattern linking L4b to L1d. On the other hand, \Cref{fig:example-6} demonstrates the indirect decision rule applied to L1b, which is marked as \textsf{to-skip} due to the absence of a certificate pattern linking L4b to L1b.

This is the \textbf{third key design point} contributing to the safety of \sysnamec without the need for links between proposers. Namely, instead of forcing a direct happened-before relationship between proposer slots, we take advantage of the predefined total ordering of proposer slots to ensure that any decision is recursively carried forward such that no matter the commit pattern, the commit decisions are deterministic.

\para{Step 3: Commit sequence}
After processing all slots, the validator derives an ordered sequence of slots. Subsequently, the validator iterates over that sequence, committing all slots marked as \textsf{to-commit} and skipping all slots marked as \textsf{to-skip}. This iteration continues until the first \textsf{undecided} slot is encountered. \Cref{sec:security} demonstrates that this commit sequence is safe and that eventually all slots will be classified as either \textsf{to-commit} or \textsf{to-skip}.
In the example depicted in \Cref{fig:example}, the commit sequence is L1a, L1c, L1d, L2a.
\iflongversion
    \Cref{sec:example}
\else
    The long version of the paper~\cite{mysticeti-long}
\fi
provides a detailed walkthrough of the decision rule applied to the example DAG of \Cref{fig:example}.

This is the \textbf{final key design point} of \sysnamec; unlike prior work that commits everything the moment a decision rule exists, \sysnamec applies some backpressure through undecided slots to preserve safety. This, however, does not harm performance, as these undecided slots would have not even existed as possible commit candidates in prior designs.

\subsection{Choosing the number of proposer slots}
The example presented by \Cref{fig:example} assumes a number of proposer slots per round equal to the committee size.
While this choice offers the best latency under normal conditions, it may impact performance during periods of extreme asynchrony or under Byzantine attack.

In these cases, the probability that the direct decision rule fails to classify a proposer slot increases when some proposer slots are slow or equivocate. This forces the validator to resort to the indirect decision rule more often.
As a result, there can be an increase in the number of undecided slots, which in turn delays the commit sequence. \Cref{fig:example} illustrates this example through the classification of L2c and L1b as \textsf{undecided}, preventing the exemplified protocol execution from immediately committing L2d, L3b, L3c, L3d, L4b, L4c, and L4d, which would have been possible under ideal conditions. This is nevertheless an extreme case of the adversary controlling the network and some validators only to slow down the system without any actual profit.
After a decade of running blockchains in the wild, this is not something that has been witnessed, as attackers tend to attack in order to break safety and not liveness.

Nevertheless, in order to mitigate it we use
HammerHead~\cite{tsimos2023hammerhead} in order to select $2f+1$ leaders that are best performing as candidate leaders. This strikes a good balance as it does not increase the median latency and only increases the expected latency by $\frac{1}{3}$ of a delay.
\Cref{sec:algorithms} provides detailed \sysnamec algorithms that allow the number of proposer slots per round to be configurable.

\subsection{\sysnamec Algorithms} \label{sec:algorithms}
This section presents the detailed algorithms of \sysnamec. It can be skipped if a high-level understanding is sufficient.

\begin{algorithm}[t]
    \caption{Helper functions}
    \label{alg:consensus-utils}
    \footnotesize

    \begin{algorithmic}[1]
        \Procedure{GetProposerBlock}{$w$}
        \State $r_{proposer} \gets \Call{ProposerRound}{w}$
        \State $id \gets \Call{GetPredefinedProposer}{r_{proposer}}$
        \If{$\exists b \in DAG[r_{proposer}] \text{ s.t. } b.author = id$} \Return $b$ \EndIf
        \State \Return $\perp$
        \EndProcedure

        \Statex
        \Procedure{GetFirstVotingBlocks}{$w$}
        \State $r_{voting} \gets \Call{ProposerRound}{w} + 1$
        \State \Return $DAG[r_{voting}]$
        \EndProcedure

        \Statex
        \Procedure{GetDecisionBlocks}{$w$}
        \State $r_{decision} \gets \Call{DecisionRound}{w}$
        \State \Return $DAG[r_{decision}]$
        \EndProcedure

        \Statex
        \Procedure{Link}{$b_{old}, b_{new}$}
        \State \Return exists a sequence of $k\in\mathbb{N}$ blocks $b_1, \dots, b_k$ s.t. $b_1 = b_{old}, b_k = b_{new}$ and $\forall j \in [2, k]: b_j \in \bigcup_{r \geq 1} DAG[r] \wedge b_{j-1} \in b_{j}.parents$
        \EndProcedure

        \Statex
        \Procedure{IsVote}{$b_{vote}, b_{proposer}$}
        \Function{SupportedBlock}{$b, id, r$}
        \If{$r \geq b.round$} \Return $\perp$ \EndIf
        \For{$b' \in b.parents$}
        \If{$(b'.author, b'.round) = (id, r)$} \Return $b'$ \EndIf
        \State $res \gets \Call{SupportedBlock}{b', id, r}$
        \If{$res \neq \perp$} \Return $res$ \EndIf
        \EndFor
        \State \Return $\perp$
        \EndFunction
        \State $(id, r) \gets (b_{proposer}.author, b_{proposer}.round)$
        \State \Return $\Call{SupportedBlock}{b_{vote}, id, r} = b_{proposer}$
        \EndProcedure

        \Statex
        \Procedure{IsCert}{$b_{cert}, b_{proposer}$}
        \State $res \gets |\{b \in b_{cert}.parents: \Call{IsVote}{b, b_{proposer}}\}|$
        \State \Return $res \geq 2f+1$
        \EndProcedure

        \Statex
        \Procedure{SkippedProposer}{$w$}
        \State $r_{proposer} \gets \Call{ProposerRound}{w}$
        \State $id \gets \Call{GetPredefinedProposer}{r_{proposer}}$
        \State $B \gets \Call{GetFirstVotingBlocks}{w}$
        \State $res \gets |\{b \in B \text{ s.t. } \forall b' \in b.parents: b'.author \neq id \}|$
        \State \Return $res \geq 2f+1$
        \EndProcedure

        \Statex
        \Procedure{SupportedProposer}{$w$}
        \State $b_{proposer} \gets \Call{GetProposerBlock}{w}$
        \State $B \gets \Call{GetDecisionBlocks}{w}$
        \If {$|\{b' \in B: \Call{IsCert}{b', b_{proposer}}\}| \geq 2f+1$} \State \Return $b_{proposer}$
        \EndIf
        \State \Return $\perp$
        \EndProcedure

        \Statex
        \Procedure{CertifiedLink}{$b_{anchor}, b_{proposer}$}
        \State $w \gets \Call{WaveNumber}{b_{proposer}.round}$
        \State $B \gets \Call{GetDecisionBlocks}{w}$
        \State \Return $\exists b \in B \text{ s.t. } \Call{IsCert}{b, b_{proposer}}$ \& $\Call{Link}{b, b_{anchor}}$
        \EndProcedure
    \end{algorithmic}
\end{algorithm}
\begin{algorithm}[t]
    \caption{DirectDecider Algorithm}
    \label{alg:baseline-committer}
    \footnotesize

    \begin{algorithmic}[1]
        \State \texttt{waveLength} \Comment{Defaults to $3$}
        \State \texttt{roundOffset}
        \State \texttt{proposerOffset}


        \Statex
        \Procedure{TryDirectDecide}{$w$}
        \If{$\Call{SkippedProposer}{w}$} \Return $\texttt{Skip}(w)$ \EndIf \label{alg:line:baseline:skipped-proposer}
        \State $b_{proposer} \gets \Call{SupportedProposer}{w}$
        \If{$b_{proposer} \neq \perp$} \Return $\texttt{Commit}(b_{proposer})$ \EndIf
        \State \Return $\perp$
        \EndProcedure


        \Statex
        \Procedure{WaveNumber}{$r$}
        \State \Return $(r - \texttt{roundOffset}) / \texttt{waveLength}$
        \EndProcedure

        \Statex
        \Procedure{ProposerRound}{$w$}
        \State \Return $w * \texttt{waveLength} + \texttt{roundOffset}$
        \EndProcedure

        \Statex
        \Procedure{DecisionRound}{$w$}
        \State \Return $w * \texttt{waveLength} + \texttt{waveLength} - 1 + \texttt{roundOffset}$
        \EndProcedure

        \Statex
        \Procedure{GetPredefinedProposer}{$w$}
        \State $r_{proposer} \gets \Call{ProposerRound}{w}$
        \State \Return $\Call{PredefinedProposer}{r_{proposer}+\texttt{ProposerOffset}}$
        \EndProcedure
    \end{algorithmic}
\end{algorithm}
\begin{algorithm}[t]
    \caption{\sysnamec}
    \label{alg:universal-committer}
    \footnotesize

    \begin{algorithmic}[1]
        \State \texttt{committeeSize}
        \State \texttt{waveLength} \Comment{Defaults to $3$}
        \State \texttt{numOfProposers} \Comment{Set to $2$ in \Cref{sec:evaluation}}

        \Statex
        \Procedure{TryDecide}{$r_{committed}, r_{highest}$} \label{alg:line:try-decide}
        \State $sequence \gets [\;]$
        \For{$r \in [r_{highest} \text{ down to } r_{committed}+1]$}
        \For{$l \in [\texttt{numOfProposers}-1 \text{ down to } 0]$}
        \State $i \gets r \; \% \; \texttt{waveLength}$
        \State $c \gets \textsf{DirectDecider}(\texttt{waveLength}, i, l)$ \label{alg:line:direct-decider}
        \State $w \gets c.\Call{WaveNumber}{r}$
        \State $status \gets c.\Call{TryDirectDecide}{w}$ \label{alg:line:universal:try-direct-decide}
        \If{$status = \perp$} \label{alg:line:universal:direct-decide-failed}
        \State $status \gets \Call{TryIndirectDecide}{c, w, sequence}$ \label{alg:line:universal:try-indirect-decide}
        \EndIf
        \State $sequence \gets status || sequence$
        \EndFor
        \EndFor
        \State $decided \gets [\;]$
        \For{$status \in sequence$}
        \If{$status = \perp$} break \EndIf
        \State $decided \gets decided || status$
        \EndFor
        \State \Return $decided$
        \EndProcedure

        \Statex
        \Procedure{TryIndirectDecide}{$c, w, sequence$} \label{alg:line:try-indirect-decide}
        \State $r_{decision} \gets c.\Call{DecisionRound}{w}$
        \State $anchors \gets [s \in sequence \text{ s.t. } r_{decision} < s.round]$
        \For{$a \in anchors$}
        \If{$a = \perp$} \Return $\perp$ \EndIf \label{alg:line:universal:indirect-decide-failed}
        \If{$a = \texttt{Commit}(b_{anchor})$}
        \State $b_{proposer} \gets c.\Call{GetProposerBlock}{w}$
        \If{$c.\Call{CertifiedLink}{b_{anchor}, b_{proposer}}$} \label{alg:line:universal:certified-link}
        \State \Return $\texttt{Commit}(b_{proposer})$ \label{alg:line:universal:indirect-commit}
        \Else
        \State \Return $\texttt{Skip}(w)$ \label{alg:line:universal:indirect-skip}
        \EndIf
        \EndIf
        \EndFor
        \State \Return $\perp$
        \EndProcedure
    \end{algorithmic}
\end{algorithm}

\Cref{alg:consensus-utils} provides base utility functions common to many DAG-based consensus protocols~\cite{narwhal,dag-rider,bullshark}.
The function $\Call{PredefinedProposer}{\cdot}$ of \Cref{alg:baseline-committer} is a determinist leader election function, such as round robin.
\sysnamec has one type of message; the block and its validity rules are described in \Cref{sec:dag}. Every node simply proposes blocks for every round, and the validity rules make sure this happens at a beneficial pace.

\Cref{alg:universal-committer} presents the \sysnamec algorithm that is run every time a valid block is received. \sysnamec is instantiated with the following parameters:. (1) The committee size \texttt{committeeSize}. (2) The wavelength \texttt{wave\_lenght}, which the description of \Cref{sec:consensus} assumes to always equal $3$. A larger wavelength parameter increases the probability of observing a certificate pattern (\Cref{sec:dag}) over proposer slots during periods of asynchrony but increases the median latency during periods of network synchrony. (3) The number of proposer slots per round, which the example depicted by \Cref{fig:example} of \Cref{sec:consensus} assumes to equal the committee size.

The entry point of this algorithm is the procedure $\Call{TryDecide}{\cdot}$ (\Cref{alg:line:try-decide}). It operates by instantiating a \textsf{Direct Decider} (\Cref{alg:baseline-committer}) for each possible proposer slot in each round that applies the direct decision rule (\Cref{alg:line:direct-decider}). Each \textsf{Direct Decider} instance is instantiated with a round offset $\texttt{roundOffset}=r$ and a proposer offset $\texttt{proposerOffset}=l$, such that each instance operates over a unique proposer slot. These instances try to apply the direct decision rule to their proposer slot by calling the procedure $\Call{TryDirectDecide}{\cdot}$ (\Cref{alg:line:universal:try-direct-decide}). If the direct decision rule fails, \Cref{alg:universal-committer} resorts to the indirect decision rule (\Cref{alg:line:universal:try-indirect-decide}). The algorithm returns the commit sequence.

\section{The \sysnamefpc fast path protocol} \label{sec:fastpath}
For workloads necessitating consensus, the \sysnamec protocol successfully achieves a low latency bound. However, popular workloads~\cite{neiheiser2024chiron} such as asset transfers, payments or NFT minting, can be finalized before consensus, through and even lower latency fast path. This section presents \sysnamefpc that extends the consensus protocol with such consensusless transactions.

\subsection{Embedding a fast path into the DAG}
The real-world deployment of such hybrid blockchains, exemplified by Sui~\cite{sui-lutris, sui}, capitalizes on the insight that certain objects, like coins, solely access state controlled by a single party and need not undergo consensus. These objects can be finalized through a fast path utilizing reliable broadcast. Such objects are classified as having an \emph{owned object} type as opposed to the traditional \emph{shared object} type. Transactions that exclusively involve owned objects as inputs are called \emph{fast path transactions}. Two transactions \emph{conflict} if they take as input the same owned object.

In \sysnamefpc validators include transactions, and explicitly \emph{vote} for causally past transactions, in their blocks. A validator includes a transaction $T$ in its block if it does not conflict with any other transaction for which the validator has previously voted. This is also an implicit vote for the transaction. Other validators, include explicit votes for $T$ in a block $B$ if: (i) $T$ is present in the causal history of $B$; and (ii) $T$ does not conflict with any other already voted on transaction. In our implementation (\Cref{sec:implementation}), we denote the vote for a transaction $T$ appearing in block $B$ at position $i$ as the tuple $(B,i)$.
Once $T$ has $2f+1$ votes from distinct validators, we call $T$ \emph{certified}. It is a guarantee that no two conflicting transaction will be certified in the same epoch. This is the basis of the fast path safety. Transaction $T$ is finalized when either (i) there exists $2f+1$ validators supporting a certificate over $T$, even before a \sysnamec commit, or (ii)  \sysnamec commits through consensus a block that contains a certificate over $T$ in its causal history (see \Cref{sec:fastpath-execution}).

In contrast to previous approaches~\cite{sui-lutris, fastpay, zef, astro}, the fast path in \sysnamefpc is integrated within the DAG structure itself. This eliminates the need for additional protocol messages and for validators to individually sign each fast-path transaction. Instead, a validator's fast path votes are embedded within its signed blocks, which are already produced as part of the consensus protocol. Consequently, in addition to the block contents of \sysnamec, blocks in \sysnamefpc also incorporate explicit votes for transactions involving at least one owned object input. This deep embedding in the DAG additionally simplifies checkpoints~\cite{sui-lutris} as it does not require an external sub-protocol to collect all fast-path transactions that have been finalized. Instead, \sysnamefpc simply defines checkpoints as the set of finalized fast path transactions referenced by the causal history of each \sysnamec commit. These can then be used to make sure that all validators have the same state for an epoch change.

To summarize, \sysnamefpc offers several advantages compared to prior work: (i) A reduction in the number of signature generation and verification operations alleviating the compute bottleneck. (ii) Elimination of a separate post-consensus checkpointing mechanism, resulting in reduced synchronization latency, as the consensus commits themselves serve as checkpoints. (iii) Simplification of the epoch close mechanism, as we examine next.

\subsection{Execution and finality} \label{sec:fastpath-execution}
Similarly to Sui~\cite{sui-lutris}, \sysnamefpc introduces the distinction between \emph{fast path execution} and \emph{fast path finality}. The former refers to the moment when a transaction is executed by a validator, the execution effects are known, and the validator can execute subsequent transactions over the same object. The latter signifies when a transaction is considered final, ensuring persistence across epoch boundaries and validator reconfigurations.

\para{Fast path execution}
A validator can safely execute a fast path transaction once it observes blocks from $2f+1$ validators that include a vote for the transaction. Due to quorum intersection, no correct validator will ever execute conflicting fast path transactions.
\Cref{fig:fastpath-execution} illustrates a DAG pattern enabling the validator to safely execute fast path transactions $T_1$ and $T_3$. The blocks $(A_0, r, \cdot)$ contain the fast path transactions $T_1$, $T_3$, and $T_6$, while the blocks $(A_0, r+1, \cdot)$, $(A_1, r+1, \cdot)$, and $(A_2, r+1, \cdot)$ support $(A_0, r, L_r)$ and explicitly vote for $T_1$ and $T_3$ (but not for $T_6$\footnote{Transaction $T_6$ may conflict with another transaction for which the validator already voted.}). Upon observing these blocks, the validator can safely execute $T_1$ and $T_3$.
Note that \sysnamefpc transaction execution can be extremely low-latency, requiring only a single round of communication, as opposed to the $2$ rounds required by related work~\cite{sui-lutris,fastpay,astro,zef}.

\begin{figure}[t]
    \centering
    \includegraphics[scale=0.45]{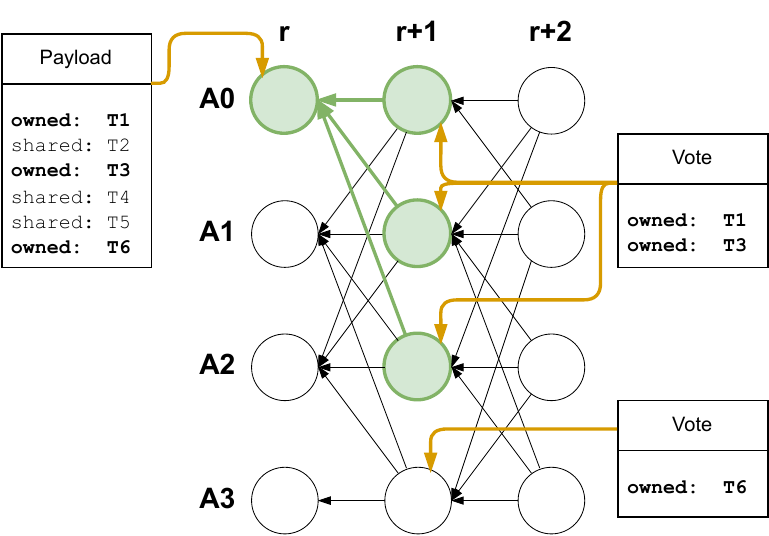}
    \caption{\footnotesize Illustration fast path transaction execution. The blocks $(A_0, r, \cdot)$ contain the fast path transactions $T_1$, $T_3$, and $T_6$. Blocks $(A_0, r+1, \cdot), (A_1, r+1, \cdot), (A_2, r+1, \cdot)$ support $(A_0, r, L_r)$ and explicitly vote for $T_1$ and $T_3$ (but not $T_6$). Upon observing these blocks, the validator can safely execute $T_1$ and $T_3$.}
    \label{fig:fastpath-execution}
\end{figure}

\para{Fast path finality}
Transactions executed by some honest validators can still be reverted since there is no guarantee that other validators will eventually observe sufficient evidence to execute the transaction.
For instance, \Cref{fig:revert-execution} illustrates a scenario where transactions $T_1$ and $T_3$ are executed by validator $A_3$ at round $r+3$, but no proposals from that validators are included into the DAG for rest of the epoch, possibly due to network asynchrony. Consequently, no other validator observes sufficient evidence to execute those transactions, and validator $A_3$ reverts their execution upon epoch change. Note that reverting execution is a straightforward operation and already supported by the Sui protocol, the only blockchain deploying a fast path.

\begin{figure}[t]
    \centering
    \includegraphics[scale=0.45]{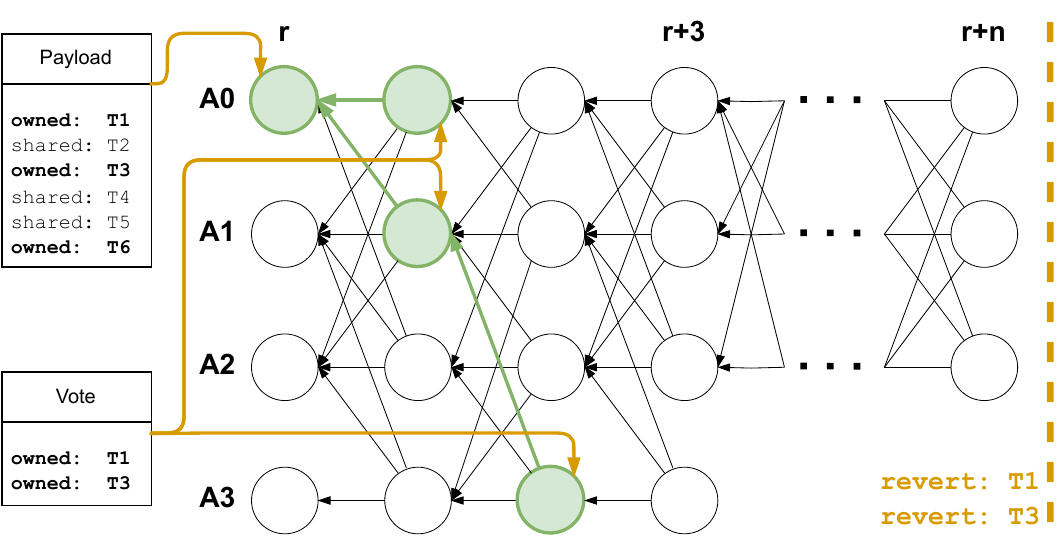}
    \caption{\footnotesize Illustration of the scenario where transactions $T_1$ and $T_3$ are executed by validator $A_3$ at round $r+2$ but no other validator observes sufficient votes to execute those transactions, and validator $A_3$ reverts their execution upon epoch change.}
    \label{fig:revert-execution}
\end{figure}

To ensure that the effects of a fast path transaction endure across epoch boundaries and validator reconfiguration, it must be \emph{finalized}. A fast path transaction is finalized when the validator observes either (1) $2f+1$ certificate patterns over the block proposing the transaction (as detailed in \Cref{sec:dag}), each containing $2f+1$ votes for the transaction, or (2) a single certificate pattern over the block proposing the transaction, which includes $2f+1$ votes for the transaction and is referenced in the causal history of a block committed by the consensus protocol.
\Cref{fig:fastpath-finalization} illustrates these two possible finality pattern for fast path transactions $T_1$ and $T_3$.

The finality of a fast-path transaction across epochs is proven by \Cref{thm:epoch-safety} of \Cref{sec:security}.
Additionally, \Cref{sec:mixed-objects} outlines how \sysnamefpc accommodates transactions containing both owned object and non-owned object inputs.

\begin{figure}[t]
    \centering
    \begin{subfigure}[t]{0.5\textwidth}
        \centering
        \includegraphics[scale=0.45]{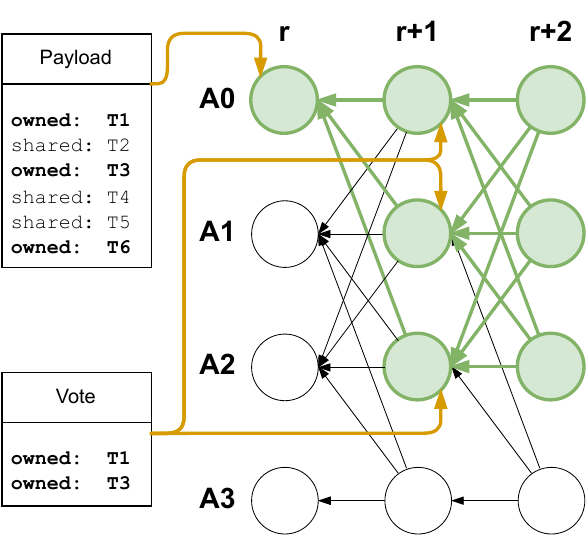}
        \caption{\footnotesize
            Transactions $T_1$ and $T_3$ proposed by $(A_0, r, \cdot)$ are finalized at round $r+2$ upon observing the $2f+1$ certificate pattern defined by $(A_0, r+2, \cdot)$, $(A_1, r+2, \cdot)$, and $(A_3, r+2, \cdot)$, referencing the $2f+1$ blocks $(A_0, r, \cdot)$, $(A_1, r, \cdot)$, and $(A_3, r, \cdot)$ that explicitly vote for $T_1$ and $T_3$.
        }
        \label{fig:fastpath-finalization-1}
    \end{subfigure}
    \begin{subfigure}[t]{0.5\textwidth}
        \centering
        \includegraphics[scale=0.45]{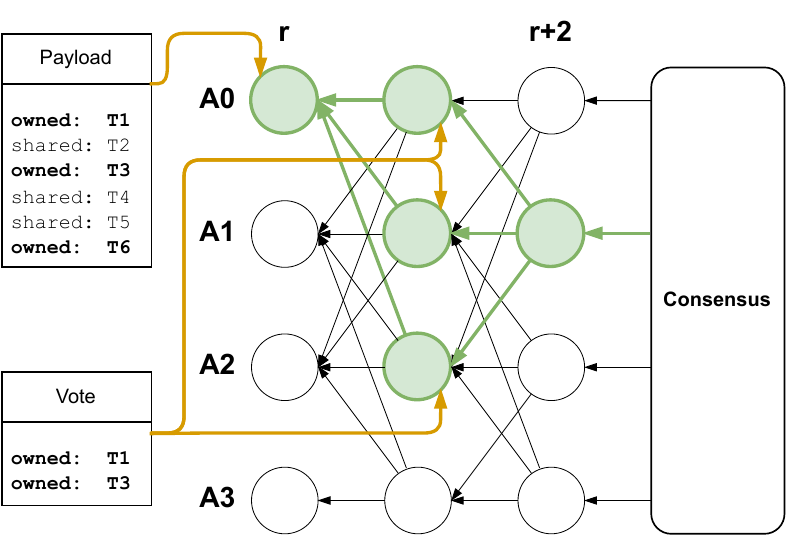}
        \caption{ \footnotesize
            Transactions $T_1$ and $T_3$ proposed by $(A_0, r, \cdot)$ are finalized after consensus upon committing block $(A_1, r+2, \cdot)$. This block defines a certificate pattern over $(A_0, r, \cdot)$ that contains $(A_0, r+1, \cdot)$, $(A_1, r+1, \cdot)$, and $(A_3, r+1, \cdot)$ that vote for $T_1$ and $T_3$.
        }
        \label{fig:fastpath-finalization-2}
    \end{subfigure}
    \caption{
        \footnotesize Illustration of the two fast path transaction finalization scenarios.
    }
    \label{fig:fastpath-finalization}
\end{figure}

\subsection{Mixed-objects transactions} \label{sec:mixed-objects}
\sysnamefpc allows for transactions that contain both owned-object and non-owned-object inputs. Such transactions are called \emph{mixed-objects transactions}. Validators execute and finalize these transactions upon observing (1) blocks from $2f+1$ validators that include a vote for the transaction, and (2) a block committed by the consensus protocol referencing these blocks in its causal history.

\Cref{fig:mixed-objects} provides an example illustrating the finalization of a mixed-object transaction. This mechanism intuitively operates in two steps: first, it ``locks'' the owned-object inputs, and then sequences this lock to prevent the execution of potentially conflicting owned-object transactions. The safety of this approach is guaranteed by \Cref{thm:fpc-safety} of \Cref{sec:security}.

\begin{figure}[t]
    \centering
    \includegraphics[scale=0.45]{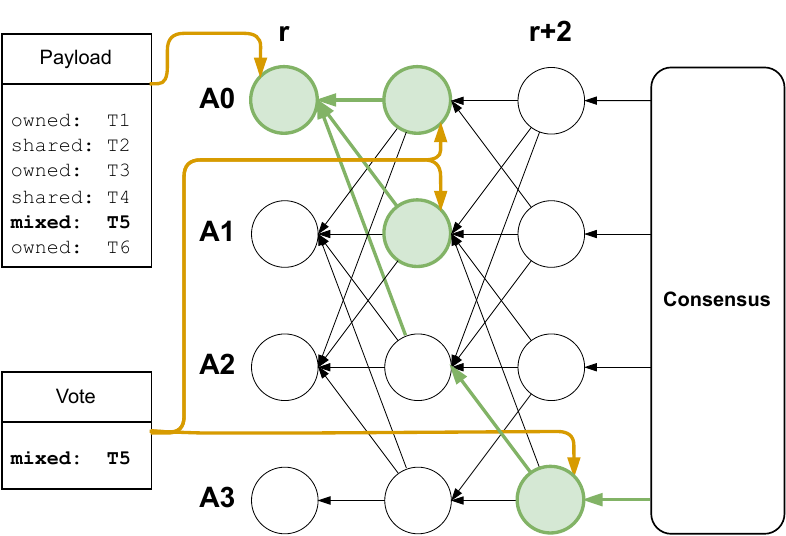}
    \caption{
        \footnotesize Illustration of a mixed-objects transaction $T_5$ that contains both owned-object inputs and non-owned-object inputs.
        $T_5$ is proposed as part of $(A_0, r, \cdot)$. Blocks $(A_0, r_1, \cdot)$, $(A_1, r+1, \cdot)$, and $(A_3, r+2, \cdot)$ vote for $T_5$. The validator can execute and finalize $T_5$ once block $(A_3, r+2, \cdot)$ is committed by the consensus protocol.
    }
    \label{fig:mixed-objects}
\end{figure}

\subsection{Epoch change and reconfiguration} \label{sec:epoch-change}

As mentioned in \Cref{sec:consensus}, quorum-based blockchains typically operate in epochs, allowing validators to join and leave the system at epoch boundaries. Moreover, epoch boundaries serve as natural boundaries for protocols with a consensusless path to ``unlock" transactions that have lost liveness due to equivocation from the client~\cite{sui-lutris, cuttlefish}.
This committee reconfiguration process must uphold a critical safety property: transactions finalized in an epoch should persist across subsequent epochs. In other words, transactions finalized in the current epoch should not conflict with transactions committed in future epochs. This holds trivially for consensus protocols which is why we omit the epoch change for \sysnamec.

\para{The \sysnamefpc epoch-change protocol}
The safety of reconfiguration is ensured by including all finalized transactions from the current epoch into the causal history of the epoch's final commit, which also acts as the initial state for the succeeding epoch. Guaranteeing reconfiguration safety is straightforward in systems mandating consensus for all transactions, such as \sysnamec, owing to the total ordering property inherent in consensus. A deterministic consensus commit $C$ sets the boundary between epochs $e$ and $e+1$. This makes sure that all transactions completed in epoch $e$ are included in and come before commit $C$.

However, designing reconfiguration mechanisms for systems with a consensusless fast path, like \sysnamefpc, presents non-trivial challenges. There is a race between finalized transactions being incorporated into consensus commits and new transactions being finalized by the fast path. Trivially closing the epoch may result in the final commit of the epoch failing to encompass all transactions finalized by the fast path, thereby violating the safety property of reconfiguration.

To solve this challenge, \sysnamefpc introduces an overriding bit called the \textit{epoch-change bit} in all its blocks. When this bit is set to 1 (default set to 0), it signifies that blocks referencing these votes do not contribute to the finalization of fast path transaction, irrespective of its causal history. Effectively, this epoch-change bit allows for the pause of the consensusless fast path of \sysnamefpc near the end of the epoch, mitigating the race condition highlighted above.

Epoch change starts at a predefined commit, often signaled by a higher-layer logic (e.g., a smart contract) indicating the readiness of the new committee to take charge. Once an honest validator detects the commencement of epoch change, it ceases to include transactions and to cast votes for any fast-path transactions. Subsequently, it sets the epoch-change bit to 1 in all its future blocks for the current epoch. Furthermore, while the validator continues to progress through rounds and participate in consensus, it stop processing and finalizing fast-path transactions. Upon committing blocks from $2f+1$ validators with the epoch-change bit set via the consensus path, the epoch is considered closed.

Once the epoch ends, any validator participating the committee of the next epoch may unlock fast-path transactions that were blocked due to client equivocations. These transactions can then receive fresh votes in subsequent epochs.
%

\para{Security intuition}
The epoch-change mechanism ensures that transactions finalized in an epoch (including on the fast path before consensus) persist across all subsequent epochs, a critical safety property (more formally in \Cref{thm:epoch-safety}).
Informally, by committing $2f+1$ blocks with the epoch-change bit set, we guarantee that every transaction finalized via the fast-path
would have a certificate as part of the causal history of the epoch-change commit (due to a quorum intersection argument).
Consequently, all validators process the certificate before they end of the epoch and persist execution results across epochs.

The liveness of \sysnamefpc directly depends on the liveness of \sysnamec. Informally, if the epoch is long enough, a non-conflicting transaction will gather sufficient votes, and then be certified by $2f+1$ blocks with the epoch-change bit unset. Which in turn ensures that it will be included in a commit and persisted across epochs.
\Cref{sec:security-fpc} formally proves the safety and liveness of \sysnamefpc.

\section{\sysname Security} \label{sec:security}
\label{sec:proof}
We argue the security of \sysnamec and \sysnamefpc under the Byzantine assumption presented in \Cref{sec:overview}.

\subsection{Security of \sysnamec} \label{sec:security-c}
This section argues the safety, liveness, and integrity of \sysnamec.

\para{Safety of \sysnamec}
A validator $v_k$ broadcasts messages calling $\textit{a\_bcast}_k(b,r)$, where $b$ is a block signed by validator $v_k$ and $r$ is the block's round number, i.e., $r = b.round$.
Every validator $v_i$ has an output $\textit{a\_deliver}_i(b,b.round,v_k)$, where $v_k$ is the author of $b$ and the validator that called the corresponding $\textit{a\_bcast}_k(b,b.round)$.

\begin{restatable}{lemma}{recursionsafety} \label{lemma:recursion-safety}
    If at a round $x$, $2f+1$ blocks from distinct authorities certify a block $B$, then all blocks at future rounds ($>x$) will link to a certificate for $B$ from round $x$.
\end{restatable}
\begin{proof}
    Each block links to $2f+1$ blocks from the previous round.
    For the sake of contradiction, assume that a block in round $r (>x)$ does not link to a certificate from round $x$.
    If $r = x+1$, by the standard quorum intersection argument, a correct validator equivocated in round $x$, which is a contradiction.
    Similarly, if $r > x+1$, by the standard quorum intersection argument, a correct validator's block in round $r-1$ does not link to its own block in round $x$, which is also a contradiction.
\end{proof}

\begin{restatable}{lemma}{directskip} \label{lemma:direct-skip}
    If a correct validator commits some block in a slot $s$, then no correct validator decides to directly skip the slot $s$.
\end{restatable}
\begin{proof}
    A validator $X$ decides to directly skip a slot $s$ if there is no support during the support rounds for any block corresponding to $s$. If another validator committed some block $b$ for slot $s$, at least $f+1$ correct validators supported $b$. By the quorum intersection argument, $X$ must have observed at least one validator supporting $B$, which is a contradiction.
\end{proof}

\begin{restatable}{lemma}{directcommitsafety} \label{lemma:direct-commit-safety}
    If a correct validator directly commits some block in a slot $s$, then no correct validator decides to skip $s$.
\end{restatable}
\begin{proof}
    For the sake of contradiction, assume that a correct validator $X$ directly commits block $b$ in slot $s$ while another correct validator $Y$ decides to skip the slot.
    $Y$ can decide to skip the slot $s$ in one of two ways: (a) $Y$ directly skipped $s$ because there was no support during the support rounds for any block corresponding to $s$, or (b) $Y$ skipped $s$ during the recursive commits triggered by a direct commit of a later slot.

    Case (a). Direct contradiction of Lemma~\ref{lemma:direct-skip}.

    Case (b). Let block $b'$ denote the proposer block, committed during the recursive indirect commits, that allowed $Y$ to decide $s$ as skipped. Due to the commit rule, the round number of $b'$ is greater than the decision round of $s$, and $b'$ does not link to a certificate for $b$. Since $X$ committed $b$, there are $2f+1$ certificates for $b$ in its decision round, leading to a contradiction due to Lemma~\ref{lemma:recursion-safety}.
\end{proof}

\begin{restatable}{lemma}{uniquesupport} \label{lemma:unique-support}
    For any slot $s \equiv (v,r)$, a correct validator never supports two distinct block proposals from validator $v$ in round $r$ across all of its blocks.
\end{restatable}
\begin{proof}
    By definition, a block can only support at most a single proposal for a particular slot $s$.
    Block support is calculated through a depth-first traversal of the referenced blocks, such that the first block corresponding to $s$ encountered during the traversal is supported.
    Since a correct validator first includes a reference to its own block from the previous round, once a correct validator supports a certain block for $s$, it continues to support the same block in all of its future blocks.
\end{proof}

\begin{restatable}{lemma}{uniquecert} \label{lemma:unique-cert}
    For any slot, at most a single block will ever be certified, i.e. gather a quorum ($2f+1$) of support.
\end{restatable}
\begin{proof}
    For contradiction's sake, assume that two distinct block proposals for a slot gather a quorum of support. By the standard quorum intersection argument, a correct validator supports two distinct blocks for the same slot, which is a contradiction of the proved Lemma~\ref{lemma:unique-support}.
\end{proof}

As a result of Lemma~\ref{lemma:unique-cert}, we get the following corollary:

\begin{corollary}
    \label{corollary:unique-commit}
    No two correct validators commit distinct blocks for the same slot.
\end{corollary}

\begin{lemma}
    \label{lemma:consistent-slots}
    All correct validators have a consistent state for each slot, i.e. if two validators have decided the state of a slot, then both either commit the same block or skip the slot.
\end{lemma}
\begin{proof}
    Let $[x_i]_{i=0}^n$ and $[y_i]_{i=0}^m$ denote the state of the slots for two correct validators $X$ and $Y$, such that $n$ and $m$ are respectively the indices of the highest committed slot. WLOG $n \leq m$. Any slot decided by $X$ higher than $n$ are direct skips and are therefore consistent with $Y$ due to Lemma~\ref{lemma:direct-skip}.
    We now prove, by induction, statement $P(i)$ for $0 \leq i \leq n$: if $X$ and $Y$ both decide the slot $i$, then both either commit the same block or skip the slot.

    Base Case: $i=n$. $X$ directly commits slot $i$, the highest committed slot for $X$. From Lemma~\ref{lemma:direct-commit-safety}, if $Y$ decides slot $i$, then it must also commit slot $i$. By Corollary~\ref{corollary:unique-commit}, $Y$ commits the same block.

    Assuming $P(i)$ is true for $k+1 \leq i \leq n$, we now prove $P(k)$. Similar to the base case, if one validator decides to directly commit a block in slot $k$, then the other validator, if it also decides slot $k$, decides to commit the same block. If one validator decides to directly skip slot $k$, then the other validator, if it also decides slot $k$, decides to skip due to Lemma~\ref{lemma:direct-skip}. We now analyze the only remaining case where $X$ and $Y$ indirectly decide the slot $k$. Let $k'$ denote the first slot $>k$ with a round number higher than the decision round of $k$. There exist slots $k_x (\geq k')$ and $k_y(\geq k')$ such that $X$ commits block $b_x$ in $k_x$ while skipping all slots in $[k', k_x)]$ and $Y$ commits block $b_y$ in $k_y$ while deciding to skip all slots in $[k', k_y)]$.
    As $k_x \leq n$, it follows from the induction hypothesis that $k_x = k_y$ and $b_x = b_y = b$.
    Since the indirect decision of $X$ and $Y$ for slot $k$ depends entirely on the causal history of the same block $b$, both validators decide the slot $k$ identically.
\end{proof}

\begin{lemma}
    \label{lemma:proposer-consistent}
    All correct validators commit a consistent sequence of proposer blocks (i.e., the committed proposer sequence of one correct validator is a prefix of another's).
\end{lemma}
\begin{proof}
    The committed sequence of proposer blocks is nothing but the sequence of committed blocks before the first undecided slot.
    The statement is then a direct implication of Lemma~\ref{lemma:consistent-slots}.
\end{proof}

\begin{theorem}[Total Order]
    \label{thm:total-order}
    \sysnamec satisfies the total order property of Byzantine Atomic Broadcast.
\end{theorem}
\begin{proof}
    Correct validators deliver blocks by using an identical deterministic algorithm to order the causal history of committed proposer blocks.
    Since a correct validator has all the causal histories of a block when the block is added to its DAG, and the sequence of committed proposer blocks of one validator is a prefix of another's (Lemma~\ref{lemma:proposer-consistent}), all correct validators deliver a consistent sequence of blocks, i.e., the sequence of blocks delivered from one validator is a prefix of the sequence delivered by any other validator. The total order property of BAB immediately follows.
\end{proof}

\para{Liveness of \sysnamec}
We show the liveness of \sysnamec under partial synchrony (\Cref{sec:overview}).

\begin{lemma}[Round-Synchronization] \label{th:view-synch}
    After GST all honest parties will enter the same round within $3\cdot\Delta$.
\end{lemma}
\begin{proof}
    After GST all messages sent before GST deliver within $\Delta$. This means that if $r$ is the highest round any honest validator proposed a block for before GST, then, accounting for $2\cdot \Delta$ to request any missing parent block, every honest validator will receive the block proposal of the honest validator at $GST+3\cdot \Delta$ and also enter $r$.
\end{proof}

\begin{lemma}[Leader-Proposal] \label{th:proposer-proposal}
    After GST an honest proposer's proposal will get votes from every honest validator.
\end{lemma}

\begin{proof}
    After GST if an honest validator enters wave $w$, then it has to broadcast the last block of wave $w-1$. Within $\Delta$ the honest proposer (and every other honest party) will receive the block and adopt the parents, being able to also enter wave $w$ as they are all synchronized (Lemma~\ref{th:view-synch}). Then the honest proposer will directly propose its block. Since the timeout is set to $4 \cdot \Delta$ the proposer's block of wave $w$ will arrive before the first honest validator times out hence, every honest validator will vote for the proposer.
\end{proof}

\begin{lemma}[Sufficient Votes] \label{th:votes}
    After GST all honest validators will create a certificate for the honest proposer.
\end{lemma}

\begin{proof}
    By Lemma~\ref{th:proposer-proposal} all honest validators will vote for an honest proposer after GST. For an honest validator to propose a block at the decision round it needs to (a) get the proposal of the proposer and (b) have $2f+1$ parents. All honest validators receive the proposer proposal within $\Delta$ since the proposer is honest. Additionally once a honest validator advances to the decision round all honest validators will receive its block proposal and adopt the parents within $3\cdot\Delta$.
    Consequently, by construction, honest validators wait for $4 \cdot \Delta$ before giving up the certificate creation and will receive the votes from all honest validators witnessing a certificate
\end{proof}

\begin{lemma}
    \label{lemma:3-consecutive}
    The round-robin schedule of proposers in \sysname ensures that in any window of  $3f+3$ rounds, there are three consecutive rounds with honest primary proposers. A primary proposer is the proposer of the first slot of a round.
\end{lemma}
\begin{proof}
    There are $3f+1$ groups of three consecutive rounds.
    Due to the round-robin schedule, each of the honest validators must be the primary proposer in exactly 3 of these groups.
    As there are $2f+1$ honest validators, due to the pigeonhole principle, one group must contain $\lceil \frac{3*(2f+1)}{3f+1} \rceil$ = 3 honest proposers.
\end{proof}

\begin{lemma}
    \label{lemma:decision-liveness}
    After GST any undecided slot eventually gets decided.
\end{lemma}
\begin{proof}
    Let there be an undecided slot $s$ in round $r$.
    After GST, due to Lemma~\ref{lemma:3-consecutive}, there will eventually be an honest proposer for the first slots $s_0, s_1$ and $s_2$ of rounds $k, k+1$ and $k+2$ respectively, where $k > r$. By Lemma~\ref{th:votes}, the honest proposer's blocks will have $2f+1$ certificates and be scheduled for a commit.
    We now prove that by induction, all slots in round $\leq k-1$ get decided.
    In the base case, any undecided slots in rounds $k-3, k-2$ or $k-1$ get decided by the commits in slots $s_0, s_1$ and $s_2$ respectively, as they are the first slots higher than the respective decision rounds.
    For the induction step, any undecided slot $s$ in round $x \leq k-4$ also gets decided since $s_0$ is higher than the decision round of $x$ and there are no undecided slots between $s$ and $s_0$ (induction hypothesis).
\end{proof}

\begin{theorem}[Consensus Liveness] \label{th:c-liveness}
    After GST the proposal of an honest proposer will commit.
\end{theorem}
\begin{proof}
    By Lemma~\ref{th:votes} there will be $2f+1$ certificates for the proposer, one per honest party. By the code an honest validator tries to commit the proposer for every block they get so eventually they will get the $2f+1$ certificates. The validator schedules the block to be committed. By Lemma~\ref{lemma:decision-liveness}, all prior undecided blocks will eventually be decided, and the validator will deliver the honest proposer's block.
\end{proof}

\begin{theorem}[Agreement]
    \label{thm:agreement}
    \sysnamec satisfies the agreement property of Byzantine Atomic Broadcast.
\end{theorem}
\begin{proof}
    If a correct validator outputs $\textit{a\_deliver}_i(b,r,v_k)$, then it must have committed a sequence of proposer blocks $L = l_0,l1...l_n$ such that the deterministic algorithm to deliver blocks from the sequence $L$ delivers block $b$.
    Another correct validator $Y$ that has not delivered $b$ will eventually see a proposal $b'$ from an honest proposer in round $r' > r$ as per the proposer schedule of \sysnamec. Due to Theorem~\ref{th:c-liveness}, after GST, $Y$ will commit the proposer's block $b'$.
    Due to Lemma~\ref{lemma:proposer-consistent}, $Y$ will also commit the proposer sequence $L$ before committing $b'$. Since $Y$ follows an identical deterministic algorithm as $X$ to deliver blocks from the committed sequence of proposer blocks, it also delivers $b'$ eventually.
\end{proof}

\para{Integrity of \sysnamec}
\sysnamec guarantees integrity by construction.

\begin{theorem}[Integrity]
    \label{thm:integrity}
    \sysnamec satisfies the integrity property of Byzantine Atomic Broadcast.
\end{theorem}
\begin{proof}
    The algorithm linearizing the causal history of a committed proposer block removes any block with duplicate sequence numbers before delivering the sequence of blocks.
\end{proof}

\subsection{Security of \sysnamefpc } \label{sec:security-fpc}
We argue the safety and liveness of \sysnamefpc.

\begin{theorem}[Epoch close safety]
    \label{thm:epoch-safety}
    Transactions finalized in an epoch continue to persist in all subsequent epochs.
\end{theorem}
\begin{proof}
    It is sufficient to prove that all fast-path transactions that are considered final have one certifying block committed in the current epoch.
    For contradiction's sake, assume that the epoch closed before any certifying block for a finalized transaction $tx$ could be committed. For the epoch to close, blocks from $2f+1$ validators with the epoch-change bit set must be committed. Since $tx$ is finalized, $2f+1$ validators, by definition, publish a block that certifies the transaction.
    By quorum intersection, one honest validator $v$ published a block $B_1$ in round $r_1$ certifying transaction $tx$, whereas a block $B_2$ in round $r_2$ from $v$ with epoch-change bit set must have been committed. All blocks published by $v$ in rounds $\geq r_2$ also have the epoch-change bit set.
    Because blocks with the epoch-change bit set, by definition, do not certify any transaction, $B_1$ is necessarily published in an earlier round than that of $B_2$ (i.e. $r_1 < r_2$). $B_1$ is therefore contained in the causal history of $B_2$, and must also have been committed, which is a contradiction.
\end{proof}

\begin{theorem}[\sysnamefpc Safety]
    \label{thm:fpc-safety}
    An honest validator in \sysnamefpc never finalizes two conflicting transactions.
\end{theorem}
\begin{proof}
    Transactions that have an owned object as input require votes from $2f+1$ validators to be finalized.
    If two conflicting fast paths are finalized, an honest validator must have voted for both transactions (by quorum intersection), hence a contradiction.
    Using a similar argument, a fast path transaction does not conflict with a consensus path transaction, as the consensus path in \sysnamefpc finalizes a transaction with an owned object input only if it has votes from $2f+1$ validators.
\end{proof}

\begin{theorem}[Fast-Path Liveness] \label{th:fp-liveness}
    An honest fast-path transaction will commit after GST.
\end{theorem}
The proof is the same as consistent broadcast. We do it after GST assuming the epoch does not end. If the epoch has infinite length then we can convert all references to $\Delta$ with  ``eventually'' and the proof will work in asynchrony.
\begin{proof}
    An honest validator will submit a fast-path transaction that does not have equivocation. As a result, all honest validators will receive it after $\Delta$ and vote. These votes will appear in the DAG after at most $4\cdot\Delta$ since any round has at most duration of timeout+$\Delta$ = $3\cdot\Delta$. In the next round, every honest validator will reference the $2f+1$ votes in their DAG and execute.
\end{proof}

\begin{theorem}[Equivocation-Tolerence] \label{th:eq-tolerance}
    If a faulty validator $v_k$ concurrently called
    $\textit{r\_bcast}_k(m,q,e)$ and  $\textit{r\_bcast}_k(m',q,e)$ with $m \neq m'$ then the rest of the validators
    either \\ $\textit{r\_deliver}_i(m,q,v_k,e)$, or $\textit{r\_deliver}_i(m',q,v_k,e)$, or there is a subsequent epoch $e'>e$
    where $v_k$ is honest, calls $\textit{r\_bcast}_k(m'',q,e')$ and all honest validators $\textit{r\_deliver}_i(m'',q,v_k,e')$,
\end{theorem}
\begin{proof}
    For the case that validators $\textit{r\_deliver}_i(m',q,v_k,e)$ it is a direct result of Theorem~\ref{th:fp-liveness}. Otherwise, from the code of the epoch change when the epoch ends all validators forget the locks they have taken on messages without certificates. As a result in a future epoch $e'$ where $v_k$ is honest and does not equivocate it will be able to commit $m$ again from Thereon~\ref{th:fp-liveness}.
\end{proof}

\section{Implementation} \label{sec:implementation}
We implement a networked multi-core \sysname validator in Rust.
It uses \texttt{tokio}~\cite{tokio}
for asynchronous networking, utilizing TCP sockets for communication without relying on any RPC frameworks. For cryptographic operations, we use \texttt{ed25519-consensus}~\cite{ed25519-consensus}
for asymmetric cryptography and \texttt{blake2}~\cite{rustcrypto-hashes}
for cryptographic hashing.
To ensure data persistence and crash recovery, integrate a Write-Ahead Log (WAL), seamlessly tailored to our specific requirements. We have intentionally avoided key-value stores like RocksDB~\cite{rocksdb}
to eliminate associated overhead and periodic compaction penalties. Our implementation optimizes I/O operations by employing vectored writes~\cite{writev}
for efficient multi-buffer writes in a single syscall. For reading the WAL, we make use of memory-mapped files while carefully minimizing redundant data copying and serialization. We use the \texttt{minibytes}~\cite{minibytes}
crates to efficiently work with memory-mapped file buffers without unsafe code.

While all network communications in our implementation are asynchronous, the core consensus code runs synchronously in a dedicated thread. This approach facilitates rigorous testing, mitigates race conditions, and allows for targeted profiling of this critical code path.

In addition to regular unit tests, we have two supplementary testing utilities. First, we developed a simulation layer that replicates the functionality of the \texttt{tokio} runtime and TCP networking. This simulated network accurately simulates real-world WAN latencies, while our tokio runtime simulator employs a discrete event simulation approach to mimic the passage of time. Utilizing this simulator, we can test a wide range of scenarios on a single machine and accurately estimate resulting latencies. It's worth noting that we've found these simulated latencies, such as commit latency, to closely mirror those observed in real-world cluster testing, provided that the cross-validator latency distribution in the simulated network is correctly configured. Second, we created an a command-line utility (called `orchestrator') designed to deploy real-world clusters of \sysname with machines distributed across the globe. The simulator has proven indispensable in identifying correctness defects, while the orchestrator has been instrumental in pinpointing performance bottlenecks.
We open-source our \sysname implementation, its simulator, and orchestration utilities\footnote{\codelink}.
\section{Evaluation} \label{sec:evaluation}
We evaluate the throughput and latency of \sysname through experiments on AWS to show its performance improvements over the state-of-the-art.

\iflongversion
      Despite the large number of BFT consensus protocols~\cite{fino,sailfish,chen2023resilient,cohen2022proof,dispersedledger,dumbo-ng,shoal,bbca-ledger,bbca-chain,chan2020streamlet}, we
\else
      We
\fi
opt to compare \sysnamec with vanilla HotStuff~\cite{DBLP:conf/podc/YinMRGA19}, HotStuff-over-Narwhal (called \emph{Narwhal-HotStuff})~\cite{narwhal}, and Bullshark~\cite{bullshark}. We select these protocols for the availability of open-source implementations and detailed benchmarking scripts, their similarity to \sysname, and their adoption in real-world deployments. We specifically select the Jolteon~\cite{jolteon} variant of HotStuff as it has been adopted by Flow~\cite{flow}, Diem~\cite{baudet2019state}, Aptos~\cite{aptos}, and Monad~\cite{monad}. We also select the Narwhal-HotStuff variant as it operates on a structured DAG as \sysname and is the most performant variant of HotStuff. We finally select Bullshark as it is a performant DAG-based protocol adopted by the Sui blockchain~\cite{sui, sui-lutris}, Aleo~\cite{aleo}, and Fleek~\cite{fleek}. We evaluate the Narwhal-based systems (that is, Narwhal-HotStuff and Bullshark) in their default 1 worker configuration.
We also evaluate the fast path \sysnamefpc against Zef~\cite{DBLP:journals/corr/abs-2201-05671} (in its default configuration, with 10 shards), which is the state-of-the-art fast path protocol that serves as the foundation for the Linera blockchain~\cite{linera}.

Throughout our evaluation, we particularly aim to demonstrate the following claims.
\textbf{C1:} \sysnamec has higher throughput and drastically lower latency than the baseline state-of-the-art protocols.
\textbf{C2:} \sysnamec has a similar throughput to the baseline protocols but maintains sub-second latencies when operating in the presence of crash faults.
\textbf{C3:} \sysnamefpc maintains the same latency as the baseline state-of-the-art consensus-less protocol but with drastically higher throughput.

Note that evaluating BFT protocols in the presence of Byzantine faults is an open research question~\cite{twins}, and state-of-the-art evidence relies on formal proofs of safety and liveness (which we present in \Cref{sec:security}). While there is a need to robustly tolerate Byzantine faults, we note that they are rare in observed delegated proof of stake blockchains, as compared to crash faults that are very common.

\subsection{Experimental setup}
We deploy a \sysname testbed on AWS, using \texttt{m5d.8xlarge} instances across 13 different AWS regions\xspace
\iflongversion
      : N. Virginia (us-east-1), Oregon (us-west-2), Canada (ca-central-1), Frankfurt (eu-central-1), Ireland (eu-west-1), London (eu-west-2), Paris (eu-west-3), Stockholm (eu-north-1), Mumbai (ap-south-1), Singapore (ap-southeast-1), Sydney (ap-southeast-2), Tokyo (ap-northeast-1), and Seoul (ap-northeast-2)
\fi.
Validators are distributed across those regions as equally as possible. Each machine provides 10Gbps of bandwidth, 32 virtual CPUs (16 physical cores) on a 2.5GHz Intel Xeon Platinum 8175, 128GB memory, and runs Linux Ubuntu server 22.04.
\iflongversion
      We select these machines because they provide decent performance, are in the price range of `commodity servers', and are the same instance types used by our baselines.
\fi

\sysname can employ more than one slot per round to mitigate the performance impact of crash faults and commit more blocks per round, but if the proposer slot behaves in a Byzantine manner, it can still manipulate their slot to remain undecided, resulting in similar latency effects as an unmasked crash fault. Therefore, we have chosen to have two proposer slots per round as an effective compromise for our experiments.
To implement the partial synchrony assumption, validators wait up to 1 second to receive a proposal from the first proposer slot of the previous round.

In the following graphs, each data point is the average latency and the error bars represent one standard deviation (error bars are sometimes too small to be visible on the graph). We instantiate several geo-distributed benchmark clients within each validator submitting transactions at a fixed rate for a duration of several minutes.
\iflongversion
      We experimentally increase the load of transactions sent to the systems, and record the throughput and latency of commits. As a result, all plots illustrate the `stead state' latency of all systems under low load, as well as the maximal throughput they can serve after which latency grows quickly.
\fi
Transactions in the benchmarks are arbitrary and contain 512 bytes. The ping latency between the validators varies from 50ms to 250ms.

When referring to \emph{latency}, we mean the time elapsed from when the client submits the transaction to when the transaction is committed by the validators. When referring to \emph{throughput}, we mean the number of committed transactions over the duration of the run.
\ifpublish
      \iflongversion
            \Cref{sec:tutorial} provides a tutorial to reproduce our experiments.
      \fi
\fi

\subsection{Benchmark in ideal conditions}

\begin{figure*}[t]
      \centering
      \includegraphics[width=\textwidth]{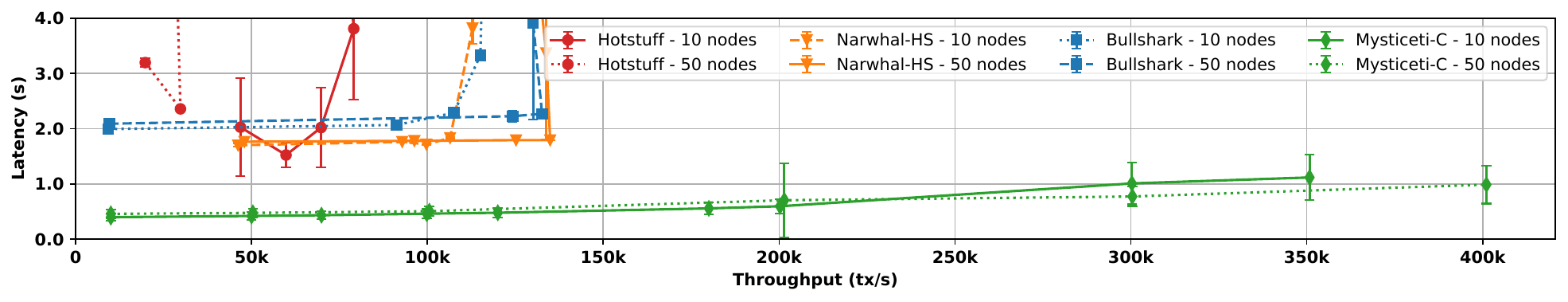}
      \caption{ \footnotesize Throughput-Latency graph comparing \sysnamec performance with state-of-the-art consensus protocols.}
      \vspace{-0.5cm}
      \label{fig:consensus}
\end{figure*}

\Cref{fig:consensus} illustrates the Latency (seconds) - Throughput (Transactions per second, TPS) relationship for \sysnamec and other consensus protocols, for a small deployment of 10 validators and a larger deployment of 50 validators. The systems run in ideal conditions, without faults.

At a steady state of 50k to 400k TPS for both network sizes \sysnamec exhibits sub-second latency, a factor 2x-3x lower than the fastest protocols, namely HotStuff, and Narwhal-HotStuff. Bullshark uses a certified DAG and worker architecture and is over 3x slower in terms of latency compared with \sysnamec for low system loads.
In terms of throughput, both \sysnamec networks scale extremely well and achieves a throughput of over 300k-400k TPS before the latency reaches 1s, that is, well lower than the latency of state-of-the-art systems. This illustrates that the single-host throughput efficiency of \sysnamec is higher than for previous designs. Note that current real-world blockchains combined\footnote{Estimates from \url{https://app.artemis.xyz/comparables}} process fewer than 100M transactions per day, equivalent to about 1.2k TPS, well within the steady state low-latency parameter space for \sysnamec, without any further scaling strategies.

These observations validate our claim \textbf{C1} showing that \sysnamec has higher throughput and drastically lower latency than the baseline state-of-the-art protocols.

Throughout these benchmarks, the the CPU utilization of the validators remains below 10\% and the validators consumes less than 15GB of memory (when experiencing the highest load of 400k tx/s).

\subsection{Benchmark with faults} \label{sec:faults}

\begin{figure}[t]
      \centering
      \includegraphics[width=\columnwidth]{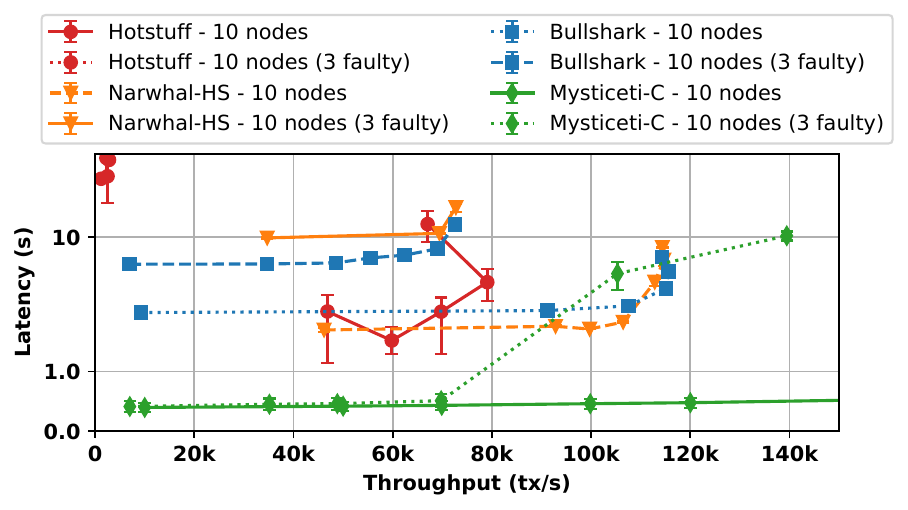}
      \caption{ \footnotesize Throughput - Latency under crash faults.}
      \label{fig:faults}
      \vspace{-0.5cm}
\end{figure}

\Cref{fig:faults} illustrates the performance of HotStuff, Narwhal-HotStuff, Bullshark, and \sysnamec when a committee of 10 parties suffers 0 to 3 crash faults (the maximum that can be tolerated in this setting). HotStuff suffers a massive degradation in both throughput and latency. With 3 faults, the throughput of HotStuff drops to a few hundred TPS and its latency exceeds 15s. Narwhal-HotStuff, Bullshark, and \sysnamec maintain a good level of throughput: the underlying DAG continues collecting and disseminating transactions despite the faults. Narwhal-HotStuff and Bullshark can process about 70k TPS in about 8-10 seconds. In contrast, \sysnamec can process the same load while maintaining sub-second latency. This improvement is due to the ability of \sysname to operate with multiple leaders per round. \sysnamec thus demonstrates a 15-20x latency improvement compared to the baseline state-of-the-art protocols.

These observations validate our claim \textbf{C2} showing that \sysnamec handles similar throughput to the state-of-the-art but with sub-second latency despite crash faults.

\subsection{Benchmark of the fast path}

\begin{figure}[t]
      \centering
      \includegraphics[width=\columnwidth]{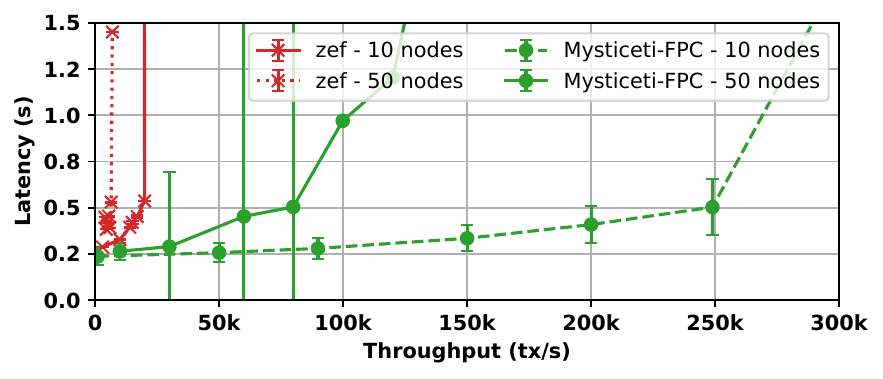}
      \caption{\footnotesize Throughput - Latency comparison the for fast path commits between \sysnamefpc and Zef}
      \label{fig:fast-path}
      \vspace{-0.5cm}
\end{figure}

\Cref{fig:fast-path} illustrates the Latency - Throughput of fast path commits for \sysnamefpc, compared with Zef~\cite{DBLP:journals/corr/abs-2201-05671} when deployed without privacy protections\footnote{Zef can also be instantiated to leverage the Coconut threshold credentials system~\cite{coconut} to provide privacy guarantees at the cost of performance.}. Both systems run in ideal conditions, without faults.
We observe that for low loads both protocols have a comparable latency of around 0.25s. However, as the load increases a Zef host has to verify and produce an increasing number of signatures, proportional to the throughput times the number of validators. As a result throughput tops at 20k TPS for a small Zef network and 7K TPS for a larger network, at a latency of 0.5s. \sysnamefpc avoids the need for individual signature verification for each transaction. At a low load, its latency is similar to Zef at 0.25s. However, as the load increases \sysnamefpc can process many more messages on a single host, namely 175k TPS for a small network and 80K for a larger network, at a latency of less than 0.5s. This is a single host throughput improvement of 8x-10x compared with Zef. We acknowledge that the Zef design can scale by adding additional hosts per validator, and sharding. However, this leads to additional hardware cost meaning that \sysnamefpc is an order of magnitude more resource efficient for the same latency.

We thus validate our claim \textbf{C3} showing that \sysnamefpc offers the same latency as state-of-the-art consensus-less protocols but with significantly higher throughput.
\section{\sysname in Production} \label{sec:production}
We collaborated with the \aquarium
team to integrate \sysnamec into the \aquarium blockchain as a replacement for Bullshark~\cite{bullshark}, which it used for consensus (\Cref{fig:switch}).

There are a number of reasons why \aquarium is a good fit for using \sysnamec. First, \aquarium maintains a fixed committee consensus during each epoch, which does not require \sysnamec to support unscheduled reconfiguration, allowing for a drop in replacement of the consensus component. Secondly, Byzantine behavior in \aquarium is handled through shifts in stake delegation between epochs. Thus, the priority is to maintain performance under frequent crash faults, as is the case with \sysnamec. Byzantine faults need to be handled safely, but it is not critical to maintain extremely high performance while doing so, since they are rare. In the past year, no Byzantine faults involving equivocation have been observed on the \aquarium mainnet.

\para{Code adaptations}
To ensure seamless integration with the existing \aquarium codebase, we undertook a series of adaptations.
We improved system resilience through the addition of new unit tests, crash recovery mechanisms, and bulk synchronization.
\iflongversion
    We further integrated into \sysnamec key \aquarium features, such as a timestamp service
    \ifpublish
        ~\cite{sui-timestamps}
    \fi
    (detailed in \Cref{sec:timestamps}), facilitating compatibility and interoperability with the higher-level smart contract layer.
\fi
Finally, we integrated HammerHead~\cite{tsimos2023hammerhead}, adding proposer reputation, to further enhance stability and performance.

\para{From prototype to production}
The roadmap spans from the initial experimentation on the prototype code to a production-ready version of \sysnamec deployed in \aquarium.

Explorations on how to integrate \sysnamec started in November 2023, with experimentation on the prototype code (\Cref{sec:implementation}) and an exploration of which existing \aquarium code components could be reused.
We reached a significant milestone in February 2024: deploying a production-ready version of \sysnamec onto the geo-distributed private test environment. Initial testing was conducted on a testbeds comprised of 137 validators with voting power that emulated the distribution observed on the \aquarium mainnet.
We run stress tests that simulate typical blockchain traffic, ranging from 100 to 6,000 transactions per second.

\iflongversion
    While the deployment went smoothly, the system's latency profile initially did not meet expectations. Although \sysnamec had sub-second P50 latency, it has noticeably higher tail latency and  jitters. We discovered several contributing factors including inefficient synchronization of unevenly broadcasted blocks, and significant degradation under high CPU utilization in the QUIC networking implementation we used. The synchronization issue was not noticeable in the current \aquarium blockchain due to the higher latency of Bullshark. In an effort to address these challenges, we improved the fault tolerance and effectiveness of block synchronization. Also we migrated to TCP networking for node-to-node communication, with streaming of proposed blocks.
\fi

Thorough testing is essential to gain confidence before the deployment of \sysnamec into the \aquarium mainnet. First, we developed and open-sourced a Domain-Specific Language (DSL) to swiftly construct \sysname's DAG \footnote{\dslcodelink} under various scenarios, which simplifies the creation of diverse DAG structures such as missing proposers, and diverse selection of block ancestors. Also we used the deterministic simulated testing framework built for the \aquarium project to randomly inject network faults, network latency jitters and thread delays to \sysnamec. The system was verified to maintain safety and liveness under these randomly injected failures. In addition, many experiments with \sysnamec were carried out over the private testnet environment with high generated load, and sometimes down validators. We made sure the system stay live and latency growth is expected under these conditions.

When deploying \sysnamec to \aquarium staging environments, the environments alternated between running Bullshark, the existing consensus protocol, and \sysnamec with each epoch. Devnet epochs last 1 hour, while testnet epochs last 24 hours. This method ensures that both consensus protocols get test coverage in the staging environments, and provides reassurance that protocol upgrades can be executed smoothly upon transitioning to the mainnet. Moreover, it allows for performance comparison between the consensus protocols under the same environment. \aquarium mainnet validators operated by independent entities voted to switch to \sysnamec on July 25th, 2024 PST.

\para{Performance assessment}
The performance results depicted in \Cref{tab:production} are provided by the \aquarium team. Measurements are obtained from a private deployment on Vultr~\cite{vultr}, utilizing \texttt{vbm-24c-256gb-amd} instances deployed on 9 different regions: Amsterdam, Frankfurt, Paris, Los Angeles, San Jose (California), Newark (New York), Tokyo, New Delhi, and Johannesburg. Each machine provides 25Gbps of bandwidth, 48 virtual CPUs (24 physical cores) on a 2.85GHz AMD EPYC 7443P, 256GB memory, and runs Linux Ubuntu server 24.04.
The partially synchronous assumption is implemented by mandating validators to wait an additional 250ms for the block of the anchor slot of the previous round after receiving $2f+1$ proposals from that previous round.

\aquarium equipped with the production-ready implementation of \sysnamec demonstrates superior latency compared to when equipped with the production-ready implementation of Bullshark, with p50 and p95 latency of 650ms and 975ms for 137 validators, respectively. In contrast, \aquarium equipped with Bullshark exhibits a p50 and p95 latency of 2.89s and 4.6s for the same configuration. The measurements are taken while both systems run in their steady-state, with a load of 5,000 transactions per second (and exhibiting an equal throughout) for multiple hours. These results demonstrate the substantial latency improvements -- of over 4x -- brought to the blockchain when swapping Bullshark for \sysnamec.

\begin{table}[t]
    \centering
    \begin{tabular}{lrrrr}
        \toprule

        \footnotesize  Protocols & \footnotesize Committee Size & \footnotesize   TPS & \footnotesize P50 Latency        & \footnotesize P95 Latency            \\
        \midrule
        \footnotesize Bullshark  & \footnotesize 137    & \footnotesize 5,000   & \footnotesize 2,890ms                 & \footnotesize 4,600ms                \\
        \footnotesize \sysnamec  & \footnotesize 137       & \footnotesize 5,000   & \footnotesize 650ms & \footnotesize 975ms \\
        \bottomrule
    \end{tabular}
    \caption{\footnotesize Comparison of production performance: bullshark vs. \sysnamec deployment within \aquarium with 137 validators (with equal voting power). Both systems are subjected to a load of 5,000 TPS and observed a sustained throughput of 5,000 TPS. All benchmarks ran for many hours.}
    \label{tab:production}
    \vspace{-0.5cm}
\end{table}

\section{Related Work}
\sysname is a family of protocols designed to support next-ge\-ne\-ra\-tion distributed ledgers~\cite{sok-consensus, chainspace, omniledger, byzcuit}. To this end, its goal is to capture as wide a range of distributed ledgers as possible whether consensus-based or consensus-less. The pioneer on hybrid distributed ledgers is the Sui Lutris blockchain~\cite{sui-lutris} which has been productionized by Sui~\cite{sui}. However, the design of Sui Lutris focuses on providing a glue between the two distinct use-cases of consensus-based and consensus-less distributed ledgers, or in the production code a glue of FastPay~\cite{fastpay} and Bullshark~\cite{bullshark}. This design process of starting with the to-be-glued components and ending in a final system has led to significant inefficiencies such as multiple rebroadcasting of the same data as well as signature verification costs.
Unlike Sui Lutris, \sysname is designed from first principles and as a result shows a potential halving of the latency, matching the lower bounds of PBFT~\cite{pbft} for consensus and Reliable Broadcast~\cite{DBLP:journals/jacm/BrachaT85,cachin2011introduction} for consensusless distributed ledgers with equivocation tolerance.

We already discussed the core benefits of \sysnamefpc in terms of much lower CPU cost. In addition, it inherits the ability to change epochs, reconfigure the validator set, and tolerate equivocations from Sui Lutris. These benefits can also be used to embed other broadcast-based protocols like FastPay~\cite{fastpay}, Astro~\cite{astro}, and Zef~\cite{DBLP:journals/corr/abs-2201-05671}, to improve privacy.

In terms of consensus, the most recent DagRider~\cite{dag-rider}, Narwhal-Tusk~\cite{narwhal}, Bullshark~\cite{bullshark} were the inspiration for using a structured DAG and defining a safe commit rule on it. However, they all use a DAG of certified blocks which increases both latency and implementation complexity. Although at first glance, certification seems to benefit adversarial cases where nodes can advance the DAG without needing to synchronize, our production experience of Bullshark~\cite{bullshark} has shown that this benefit is negated right after consensus is finished and executing transactions starts (which requires all dependencies to be already executed). As a result, the certification benefits only Byzantine Atomic Broadcast protocols but not if used for the common case of powering a State Machine Replication system (e.g., a blockchain). \sysname uses instead a DAG of signed but not certified blocks, reducing latency significantly being the fastest DAG-based SMR to date.

Cordial Miners~\cite{DBLP:conf/wdag/KeidarNPS23} has also proposed a similar DAG-structure to \sysnamec. However, their \emph{Blocklace} detects and excludes equivocating miners so that it can eventually converge when there is no misbehavior. Its expected latency is additionally higher than \sysnamec as it only commits one proposed block per wave (3 rounds) and it lacks an implementation for us to do some more direct comparison\footnote{
The Cordial Miners manuscript publicly available during \sysname's development considered a single certificate pattern to be a sufficient condition to commit a block. This is not safe. As we saw in our proofs, there is a need for $2f+1$ blocks certify a block to safely commit it. The published version of the work, that appeared concurrently to this work, fixes this issue.
}. \sysname in comparison has additionally shown how to integrate a fast path as well as how to commit most of the blocks with an expected latency of $3$ rounds.
The subsequent concurrent work on Flash~\cite{DBLP:journals/corr/abs-2305-03567} also discussed how to leverage a blocklace/DAG to allow for payments akin to the \sysnamefpc fast path, but without integrating it with a consensus path for complex transactions.
\iflongversion
    Motorway~\cite{motorway} uses a consensus protocol based on `data lanes', a relaxed notion of a DAG where replicas independently and concurrently disseminate data. Consensus is achieved over metadata through any black-box consensus mechanism, leveraging these data lanes as a data dissemination layer. This approach follows the spirit of Narwhal-HotStuff~\cite{narwhal}, with the added benefit of increased performance by eliminating the need for validators to progress in strict DAG rounds. Unfortunately, this approach also increases significantly the engineering complexity and foregoes the robustness properties that structured DAG-based protocols have.
\fi

As far as the \sysnamec commit rule is concerned, the first proposal of having a pipelined and multi-proposer version for quorum-based consensus comes from Multi-Paxos~\cite{lamport2001paxos}.
This work has been studied extensively as well as extended to multiple directions~\cite{DBLP:conf/sosp/TennageBKSJEF23,tennage2024racs,moraru2012egalitarian}. However, it only addresses crash and omission faults.
The core idea can directly be transferred to Byzantine faults as PBFT~\cite{pbft} uses a similar structure to Paxos, and we can see its adoption in Mir-BFT~\cite{stathakopoulou2022solution}. Blockmania~\cite{danezis2018blockmania} as well as Schett \& Danezis~\cite{schett2021embedding} further develop the idea for DAG-based consensus, and the recent work Shoal~\cite{shoal} has applied it to certified DAGs with recursive commit rules~\cite{bullshark}. \sysname's commit rule is the next evolution, extending pipelining into uncertified recursive DAGs in order to achieve simultaneously the lowest latency possible ($3$ message rounds, according to~\cite{DBLP:journals/tdsc/MartinA06}) as well as the high throughput and censorship resistance of DAGs.

Notably, Narwhal-based designs use a worker-primary architecture to increase throughput. \sysnamec can be adapted to this architecture, by acting as a primary for any number of workers in case additional throughput is needed.
Additionally, Shoal and HammerHead~\cite{tsimos2023hammerhead} propose leader reputation protocols inspired by Carousel~\cite{cohen2022aware}. Our production implementation of \sysnamec adopts these designs to select more reliable proposers (\Cref{sec:production}), but for liveness, it would need to adopt a proposer slot rotation schedule where slots remain static for $3$ rounds.

Previous consensus protocols such as Hashgraph~\cite{baird2016swirlds} also use a DAG of signed but not certified blocks: however, they use DAGs that are not structured as threshold clocks~\cite{DBLP:journals/corr/abs-1907-07010} making their proofs of safety very complex and leaving several open questions regarding practical implementations~\cite{narwhal}. Fino~\cite{fino} generalizes the commit rule of Bullshark to an unstructured certified DAG. BBCA-ledger~\cite{bbca-ledger} interweave together a novel low-latency happy path based on a variant of Byzantine Consistent Broadcast and Bullshark as a high-throughput DAG-based fallback path.

Notably, \sysnamec works in only 3 message communication rounds, which matches PBFT, and is optimal latency~\cite{optimal-three-round,simplex} without the use of optimistic methods like Zyzzyva~\cite{zyzzyva}. This is lower than the state-of-the-art Jolteon~\cite{jolteon} currently deployed in multiple blockchains~\cite{flow, diem, aptos, monad}. The reason is that these protocols focus on linear communication complexity, whereas \sysnamec embraces its cubic cost and amortizes it using the DAG structure as first proposed by Dag-Rider and Narwhal.
\iflongversion
  
\section{Conclusion}
We introduce \sysnamec, a threshold clock-based Byzantine consensus protocol with the lowest WAN latency of 0.5s and the ability to process over 200k TPS at this latency for single-host nodes, far exceeding the needs of blockchains today (which consume in total about 1.2k TPS). We additionally present \sysnamefpc, a fast path protocol achieves even lower latency at 0.25s but with over 8x better resource efficiency compared with protocols with explicit certificates. Despite being designed in a BFT setting, both \sysname protocols efficiently handle crash faults using multiple proposer slots per round, implemented through a novel decision rule.

\iflongversion
    We leave several explorations for the future. For use cases requiring higher throughput, we note that \sysnamec can be augmented with workers, in a similar way to Tusk and Bullshark. This would allow it to scale without known bounds, at the cost of additional latency (a round trip) to coordinate workers and primaries. An alternative approach would be to run multiple \sysnamec instances in parallel, something we feel is under-explored but inspired us to have explicit votes in \sysnamefpc. The structure of \sysnamefpc has all nodes timestamping transactions through their votes and may be useful for implementing MEV protections.

    Finally, we note that as the latency of consensus under low load shrinks (now 0.5s) the latency advantages of the fast path diminish. It is an open industrial question whether use cases that require low latency justify the complexity of dual path systems going forward, as the latency gap closes.
    It is also an open question whether the worker-primary architecture employed by Narwhal-based designs is useful today since a single worker throughput by far exceeds the capacity needed by blockchains, and does so at a lower latency. This may change in the future as more capacity is needed.
\fi

\fi

\ifpublish
  \section*{Acknowledgment}
  This work was conducted while Kushal Babel was interning with Mysten Labs. 
We would like the thank Dmitry Perelman, Xun Li, and Lu Zhang from the Mysten Labs engineering team for the great discussions that improved this work. 
We also extend our thanks to Nikita Polianskii for various discussions around potential fast path designs and to Matthias Hallgren for suggesting to dynamically adjust the number of proposer slots per round based on the network's observed behavior.
We also thank Ehud Shapiro and Oded Naor for reviewing the manuscript and related work.
Finally, we thank Kartik Nayak, Aniket Kate, and Nibesh Shrestha for identifying that our timeout for liveness needs to be 3$\Delta$.
\fi

\bibliographystyle{IEEEtran}
\bibliography{references}

\iflongversion
  \appendices
  \crefalias{section}{appendix}

  \section{Example of \sysnamec Execution} \label{sec:example}
This section completes \Cref{sec:consensus} by providing a step-by-step example of a \sysnamec execution by leveraging \Cref{fig:example-1} of \Cref{sec:consensus}. This figure illustrates an example of a \sysname DAG with four validators, (A0, A1, A2, A3), four slots per round. Initially, all proposers are marked as \textsf{undecided}.

\para{Slots classification}
The validator applies the direct decision rule starting with the latest slots, L6d, L6c, L6b, L6a, L5d, L5c, L5b, L5a (in that order), but fails to determine their status due to the absence of both a skip pattern and a certificate pattern. They thus remain \textsf{undecided}.

The direct decision rule then successfully marks L4d as \textsf{to-commit} due to the presence of $2f+1$ certificate patterns, colored in green in \Cref{fig:example-2}. This reasoning is then applied successively to L4c and L4b, also marked as \textsf{to-commit}.
\Cref{fig:example-3} then demonstrates the direct decision rule applied to L4a, resulting in its classification as \textsf{to-skip} due to the presence of a skip pattern.
Continuing with L3d, L3c, L3b, and L2d, the direct decision rule categorizes them all as \textsf{to-commit}, similar to L4d, L4c, L4b.

Moving to L2c, \Cref{fig:example-4} shows the direct decision rule failing to classify it. Lacking both a skip and certificate pattern, the validator resorts to the indirect decision rule. It first identifies the \emph{anchor} of L2c, which is the block with the lowest rank and round number $r'$ such that $(r' > r + 2)$ and that is marked as either \textsf{undecided} or \textsf{to-commit}. In this case, L2c's anchor is L5a. Since L5a is \textsf{undecided}, L2c remains so as well. The same reasoning is applied to L2b which is also marked \textsf{undecided}.
Proceeding with L2a, the direct decision rule marks it as \textsf{to-commit}.

However, the direct decision rule cannot classify L1d. Consequently, \Cref{fig:example-5} demonstrates the application of the indirect decision rule to L1d, with L4b as its anchor (note that L4a is marked \textsf{to-skip} and thus cannot be an anchor). Since L4b is marked \textsf{to-commit}, and L1d has a certificate pattern linking to its anchor, L1d is marked \textsf{to-commit}.

After marking L1c as \textsf{to-commit} in the same way as L1d, the validator analyzes L1b. Since the direct decision rule cannot decide it, the indirect decision rule is applied with L4a as its anchor. Unlike L1d, there's no certified link from L4a to L1b, resulting in L1b being marked \textsf{to-skip} (\Cref{fig:example-6}).
Finally, L1a is marked \textsf{to-commit} via the direct decision rule.

\para{Commit sequence}
With as many proposers as possible classified as either \textsf{to-commit} or \textsf{to-skip}, the validator can establish the commit sequence. Beginning with the lowest slot, it outputs slots marked as \textsf{to-commit} while skipping those marked \textsf{to-skip}, halting once an undecided slot is encountered. The resulting commit sequence is L1a, L1c, L1d, L2a.
Eventually, the DAG will progress and the slot L5a will be classified either as \textsf{to-commit} or \textsf{to-skip}, allowing the validator to classify L2c and all subsequent slots.

\para{Choosing the number of proposer slots}
The example presented by \Cref{fig:example} assumes a number of proposer slots per round equal to the committee size.
While this choice offers the best latency under normal conditions, it may impact performance during periods of extreme asynchrony or under Byzantine attack.
In these cases, the probability that the direct decision rule fails to classify a proposer slot increases when some proposer slots are slow or equivocate. This forces the validator to resort to the indirect decision rule more often.
A practical mitigation could be to dynamically adjust the number of proposer slots per round based on the network's observed behavior. This would allow the system to maintain low latency during periods of synchrony while increasing the probability of observing certificate patterns during periods of asynchrony. To ensure all validators interpret the DAG consistently, the number of proposer slots per round should be agreed upon by all validators. This can be achieved by updating the number of proposers slots per round at commit boundaries,  following a similar deterministic approach to HammerHead~\cite{tsimos2023hammerhead}

  \section{Exposing Commit Timestamps} \label{sec:timestamps}
As mentioned in \Cref{sec:production}, the production-ready implementation of \sysnamec is equipped to expose timestamps to the higher application layer. Each \sysnamec block contains both the timestamp of its proposal and its commit. Validators incorporate their current time in each block they propose. Upon receiving a block, its timestamp undergoes validation by ensuring that the included time is greater than or equal to the timestamps of its parent blocks; otherwise, the block is rejected as invalid. Honest validators will only consider blocks as parents of their proposal if they possess past timestamps with respect to their local time, while if a block arrives with a future timestamp, a validator must wait before including or rejecting it.

Consequently, if a Byzantine validator introduces a block too far into the future, it will be rejected. To mitigate the small variations in the local clocks of validators, our implementation suspends the block in memory for a brief duration if its timestamp is only slightly ahead of the current local time.

When \sysnamec outputs a commit, it associates a timestamp with this action, termed a \emph{commit timestamp}. The commit timestamp is defined as the maximum of the timestamp(s) of the proposer block(s) of such commit and the timestamp of the previous commit. Therefore, \sysnamec commit timestamps ensure monotonic increase. It's essential to include the commit timestamp of the previous commit in this maximum calculation because successive committed blocks are not necessarily linked by a parent-child relationship, thus unable to guarantee monotonicity. This contrasts with timestamp mechanisms in existing DAG-based consensus protocols, which lack a proposer every round and can thus ensure that each committed block references the previous committed block~\cite{narwhal,dag-rider,bullshark}.
\fi

\end{document}